\documentclass[runningheads]{llncs}
\usepackage[paperheight=225mm,paperwidth=155mm,textwidth=12.2cm,textheight=19.3cm,centering]{geometry}
\usepackage{hyperref}
\usepackage{breakurl}
\usepackage{version}

\usepackage{graphicx}
\usepackage{pstricks}
\usepackage{amsmath}
\usepackage{amssymb}
\usepackage{mathrsfs}
\usepackage{tikz}

\def\text#1{\textrm{#1}}

\newrgbcolor{JensRed}{0.9 0.4 0.6}
\newrgbcolor{BothRed}{1.0 0.3 0.4}
\def\precond#1{{\vphantom{#1}}^\bullet #1}
\def\postcond#1{{#1}^\bullet}

\def\production#1{\stackrel{#1}{\longrightarrow}}

\newfont{\fsc}{eusm10 scaled 1100}      % frenchscript letters
\def\powermultiset#1{\bbbn^{#1}}
\def\implies{\Rightarrow}

\def\mathrlap{\mathpalette\mathrlapinternal}
\def\mathrlapinternal#1#2{%
  \rlap{$\mathsurround=0pt#1{#2}$}}
\def\mathllap{\mathpalette\mathllapinternal}
\def\mathllapinternal#1#2{%
  \llap{$\mathsurround=0pt#1{#2}$}}
\def\trail#1{\text{~#1}}
\def\into{\rightarrow}

\def\epsilon{\varepsilon}

\def\defitem#1{\emph{#1}}

\def\AA{\text{\it AA}}

\let\origexists\exists
\let\orignexists\nexists
\let\origforall\forall
\def\quantorSpace{}
\def\exists#1.{\quantorSpace\origexists\def\quantorSpace{\,}#1.\onespace}
\def\nexists#1.{\quantorSpace\orignexists\def\quantorSpace{\,}#1.\onespace}
\def\forall#1.{\quantorSpace\origforall\def\quantorSpace{\,}#1.\onespace}
\def\onespace#1{\let\argument=#1\ifx\onespace#1\else~\fi\argument}

\let\origmin\min
\def\min{\mathord{\origmin}}
\let\origmax\max
\def\max{\mathord{\origmax}}

\def\quireunderscore{_}
\def\quire#1{%
  \def\tmp{#1}%
  \ifx\tmp\quireunderscore%
    \def\tmp{\quireindexed_}
  \else%
    \def\tmp{\mathscr{Q}#1}
  \fi\tmp}
\def\quireindexed_#1{\mathscr{Q}_{\text{#1}}}

\newtheorem{observation}{Observation}

\newcommand{\refdf}[1]{\definitionname~\ref{df-#1}}
\newcommand{\refpr}[1]{\propositionname~\ref{pr-#1}}
\newcommand{\refthm}[1]{\theoremname~\ref{thm-#1}}
\newcommand{\refcor}[1]{\corollaryname~\ref{cor-#1}}
\newcommand{\reflem}[1]{\lemmaname~\ref{lem-#1}}
\newcommand{\reffig}[1]{\figurename~\ref{fig-#1}}

\newcommand{\refsec}[1]{Section~\ref{sec-#1}}

\makeatletter
\def\goesto{\@transition\rightarrowfill}
\def\Goesto{\@transition\Rightarrowfill}
\def\ngoesto{\@transition\nrightarrowfill}
\def\nGoesto{\@transition\nRightarrowfill}
\def\@transition#1{\@ifnextchar[{\@@transition{#1}}{\@@transition{#1}[]}}
\newbox\@transbox
\newbox\@arrowbox
\def\rightarrowfill{$\m@th\mathord-\mkern-6mu%
  \cleaders\hbox{$\mkern-2mu\mathord-\mkern-2mu$}\hfill
  \mkern-6mu\mathord\rightarrow$}
\def\Rightarrowfill{$\m@th\mathord=\mkern-6mu%
  \cleaders\hbox{$\mkern-2mu\mathord=\mkern-2mu$}\hfill
  \mkern-6mu\mathord\Rightarrow$}
\def\@@transition#1[#2]%
   {\setbox\@transbox\hbox
       {\vrule height 1.5ex depth 1ex width 0ex\hskip0.25em$\scriptstyle#2$\hskip0.25em}
   \ifdim\wd\@transbox<1.5em
      \setbox\@transbox\hbox to 1.5em{\hfil\box\@transbox\hfil}\fi
   \setbox\@arrowbox\hbox to \wd\@transbox{#1}
   \ht\@arrowbox\z@\dp\@arrowbox\z@
   \setbox\@transbox\hbox{$\mathop{\box\@arrowbox}\limits^{\box\@transbox}$}
   \ht\@transbox\z@\dp\@transbox\z@
   \mathrel{\box\@transbox}}
\makeatother

\def\alignedcaption[#1&#2]{\mbox{\scriptsize $\mathllap{#1{}}\mathrlap{#2}$}}

\def\ie{i.e.\ }

\def\restrictedto{\mathop\upharpoonright}

\def\lastandname{\ and}
\newcommand{\plat}[1]{\raisebox{0pt}[0pt][0pt]{#1}}   % no vertical space
\newcommand{\inp}{\mathbin\in}                        % less space around \in
\def\idx#1#2#3#4#5{
  \def\argone{#1}
  \def\argtwo{#2}
  \def\argthree{#3}
  \def\argfour{#4}
  \def\argfive{#5}
  \def\testprime{'}
  \def\testdprime{''}
  \def\testtprime{'''}
  {\vphantom{\argthree}}_%
    {\vphantom{\argfour}\argone}^%
    {\vphantom{\argfive}{\argtwo}}%
  \argthree_%
    {\vphantom{\argone}\argfour}%
    \ifx\argfive\testprime\argfive\else%
    \ifx\argfive\testdprime\argfive\else%
    \ifx\argfive\testtprime\argfive\else%
    ^{\vphantom{\argtwo}\argfive}\fi\fi\fi%
}

\newcommand{\monus}{\mathrel{\raisebox{-0pt}[0pt][0pt]{$
                      \stackrel{\raisebox{-5pt}[0pt][0pt]{\huge$\cdot$}}
                               {\raisebox{0pt}[0pt][0pt]{$-$}}$}}}

\def\swap{\mbox{swap}}
\def\swapeq{\equiv_1}
\def\swapeqi{\equiv_{1}}
\newcommand{\adjacent}{\equiv_0}
\def\connectedto{\adjacent^*}
\def\Lin{\text{Lin}}
\def\FS{\text{FS}}

\newenvironment{itemise}{\begin{list}{$\bullet$}{\leftmargin 12pt \labelwidth\leftmargin\advance\labelwidth-\labelsep \topsep 4pt \itemsep 2pt \parsep 2pt}}{\end{list}}
\newenvironment{itemisei}{\begin{list}{$-$}{\leftmargin 12pt \labelwidth\leftmargin\advance\labelwidth-\labelsep \topsep 4pt \itemsep 2pt \parsep 2pt}}{\end{list}}
\def\justempty{}
\newenvironment{define}[1]{\begin{definition}\rm\def\arga{#1}\ifx\justempty\arga~\else#1\fi\\\vspace{-6ex}\\\mbox{~}\begin{itemise}}{\end{itemise}\end{definition}}
\newenvironment{defineCite}[2]{\begin{definition}[#1]\rm\def\arga{#2}\ifx\justempty\arga~\else#2\fi\\\vspace{-6ex}\\\mbox{~}\begin{itemise}}{\end{itemise}\end{definition}}

\def\lastandname{~{\small\&}}           % between last two authors    
\DeclareFontFamily{T1}{la}{}
\DeclareFontShape{T1}{la}{m}{n}{<->s*[0.8571428571]la14}{}

\def\processfont#1{\text{\fsc #1}}
\def\NN{\processfont{N}}
\def\SS{\processfont{S}}
\def\TT{\processfont{T}}
\def\FF{\processfont{F}}
\def\MM{\processfont{M}}
\def\PP{P}
\def\QQ{Q}
\newcommand{\pos}{\text{pos}}

\newcommand{\GR}{{\rm GR}(N)}
\newcommand{\fGR}{{\rm GR}_{\it fin}(N)}
\newcommand{\fFS}{{\rm FS}(N)}
\newcommand{\FSinf}{{\rm FS}^\infty}
\newcommand{\wFS}{\FSinf(N)}
\newcommand{\n}{\mathrel{_n\!\!=}}



\makeatletter
\headerps@out{%
[%
  /BBox[0 0 1 1]%
  /_objdef{pdfRedHighlight} %
/BP pdfmark %
gsave 1 0 0 setrgbcolor 0.7 .setopacityalpha /Multiply .setblendmode 0.1 setlinewidth 0 0 moveto 1 0 lineto stroke grestore
[/EP pdfmark%
[%
  /BBox[0 0 1 1]%
  /_objdef{pdfGreenHighlight} %
/BP pdfmark %
gsave 0 1 0 setrgbcolor 0.7 .setopacityalpha /Multiply .setblendmode 0.1 setlinewidth 0 0 moveto 1 0 lineto stroke grestore
[/EP pdfmark%
[%
  /BBox[0 0 1 1]%
  /_objdef{pdfCyanHighlight} %
/BP pdfmark %
gsave 0 1 1 setrgbcolor 0.7 .setopacityalpha /Multiply .setblendmode 0.1 setlinewidth 0 0 moveto 1 0 lineto stroke grestore
[/EP pdfmark%
}
\def\@linkborderhighlight{pdfRedHighlight}
\def\@citeborderhighlight{pdfGreenHighlight}
\def\@urlborderhighlight{pdfCyanHighlight}
\define@key{PDF}{AP}{\pdf@addtoks{#1}{AP}}
\def\hyper@linkend{%
  \literalps@out{\strip@pt@and@otherjunk\baselineskip\space H.L}%
  \edef\Hy@tempcolor{%
    \csname @\hyper@currentlinktype borderhighlight\endcsname
  }%
  \pdfmark{%
    pdfmark=/ANN,%
    linktype=link,%
    Subtype=/Link,%
    PDFAFlags=4,%
    Dest=\hyper@currentanchor,%
    AP={%
      <</N {\Hy@tempcolor}>>%
    },%
    Border=0 0 0,%
    Raw=H.B%
  }%
  \Hy@endcolorlink
  \ifHy@breaklinks
  \else
    \Hy@LinkMath
    \Hy@SaveSpaceFactor
    \egroup
    \Hy@RestoreSpaceFactor
  \fi
}%
\def\hyper@link#1#2#3{%
  \Hy@VerboseLinkStart{#1}{#2}%
  \edef\Hy@tempcolor{\csname @#1borderhighlight\endcsname}%
  \begingroup
    \protected@edef\Hy@testname{#2}%
    \ifx\Hy@testname\@empty
      \Hy@Warning{%
        Empty destination name,\MessageBreak
        using `\Hy@undefinedname'%
      }%
      \let\Hy@testname\Hy@undefinedname
    \fi
    \pdfmark[{#3}]{%
      linktype={#1},%
      Border=0 0 0,%
      pdfmark=/ANN,%
      Subtype=/Link,%
      PDFAFlags=4,%
      AP={%
        <</N {\Hy@tempcolor}>>%
      },%
      Dest=\Hy@testname
    }%
  \endgroup
}
\def\hyper@linkurl#1#2{%
  \begingroup
    \Hy@pstringdef\Hy@pstringURI{#2}%
    \hyper@chars
    \leavevmode
    \pdfmark[{#1}]{%
      pdfmark=/ANN,%
      linktype=url,%
      Border=0 0 0,%
      Action={<<%
        /Subtype/URI%
        /URI(\Hy@pstringURI)%
        \ifHy@href@ismap
          /IsMap true%
        \fi
      >>},%
      Subtype=/Link,%
      AP={%
        <</N {\@urlborderhighlight}>>%
      },%
      PDFAFlags=4%
    }%
  \endgroup
}
\makeatother
\usepackage{pstricks}
\usepackage{pst-node}
\usepackage{graphicx}

\newcounter{netimage}
\def\p#1:#2;{\cnode #1{0.3}{n\thenetimage-#2}}
\def\P#1:#2;{\p #1:#2;\pscircle*#1{0.1}}
\def\q#1:#2:#3;{\p #1:#2;\rput#1{\rput[l](0.45,0){\large\it #3}}}
\def\Q#1:#2:#3;{\P #1:#2;\rput#1{\rput[l](0.45,0){\large\it #3}}}
\def\qq#1:#2:#3;{\p #1:#2;\rput#1{\rput[t](0,-0.5){\large\it #3}}}
\def\ql#1:#2:#3;{\p #1:#2;\rput#1{\rput[r](-0.45,0){\large\it #3}}}
\def\qt#1:#2:#3;{\p #1:#2;\rput#1{\rput[b](0,0.45){\large\it #3}}}
\def\Qt#1:#2:#3;{\P #1:#2;\rput#1{\rput[b](0,0.45){\large\it #3}}}
\def\Ql#1:#2:#3;{\P #1:#2;\rput#1{\rput[r](-0.45,0){\large\it #3}}}
\def\qx#1:#2:#3:#4;{\p #1:#2;\rput#1{\rput#4{\large\it #3}}}
\def\QXX#1:#2:#3:#4:#5;{\p #1:#2;\rput#1{\rput#4{\large\it #3}}\pscircle*#5{0.1}}
\def\s#1:#2:#3;{\p #1:#2;\rput#1{\rput(-0.03,0){\large\it #3}}}
\def\PNsr#1:#2:#3:#4;{\p #1:#2;\rput#1{\rput(-0.03,0){\large\it #3}\rput[l](0.45,0){\large\it #4}}}
\def\PNsl#1:#2:#3:#4;{\p #1:#2;\rput#1{\rput(-0.03,0){\large\it #3}\rput[r](-0.45,0){\large\it #4}}}
\def\t#1:#2:#3;{\rput#1{\rnode{n\thenetimage-#2}{\psframebox{%
  \vbox to 0.6cm{\vfil\hbox to 0.6cm{\hfil\Large\it #3\hfil}\vfil}}}}}
\def\u#1:#2:#3:#4;{\rput#1{\rnode{n\thenetimage-#2}{\psframebox{%
  \vbox to 0.6cm{\vfil\hbox to 0.6cm{\hfil\Large\it #3\hfil}\vfil}}}}%
  \rput#1{\rput[l](0.6,0){\large\it #4}}}
\def\ut#1:#2:#3:#4;{\rput#1{\rnode{n\thenetimage-#2}{\psframebox{%
  \vbox to 0.6cm{\vfil\hbox to 0.6cm{\hfil\Large\it #3\hfil}\vfil}}}}%
  \rput#1{\rput[b](0,0.6){\large\it #4}}}
\def\ul#1:#2:#3:#4;{\rput#1{\rnode{n\thenetimage-#2}{\psframebox{%
  \vbox to 0.6cm{\vfil\hbox to 0.6cm{\hfil\Large\it #3\hfil}\vfil}}}}%
  \rput#1{\rput[r](-0.6,0){\large\it #4}}}
\def\a#1->#2;{\ncline{->}{n\thenetimage-#1}{n\thenetimage-#2}}
\def\A#1->#2;{\ncarc[arcangle=16]{->}{n\thenetimage-#1}{n\thenetimage-#2}}
\def\AA#1->#2;{\ncarc[arcangle=32]{->}{n\thenetimage-#1}{n\thenetimage-#2}}
\def\B#1->#2;{\ncarc[arcangle=-8]{->}{n\thenetimage-#1}{n\thenetimage-#2}}
\def\avlinearc{0.2}
\def\av#1[#2]-#3->[#4]#5;{
  \SpecialCoor
  \psline[linearc=\avlinearc]{->}([angle=#2]n\thenetimage-#1)#3([angle=#4]n\thenetimage-#5)
}

\makeatletter
\long\def\petrinet(#1)#2\end{\psscalebox{0.7}{\pspicture(#1)\stepcounter{netimage}#2\endpspicture}\end}

\makeatother

\psset{arrowsize=5pt}

\titlerunning{Abstract Processes in the Absence of Conflicts in General P/T-Systems}
\title{Abstract Processes in the Absence of Conflicts\texorpdfstring{\newline}{}
 in General Place/Transition Systems%
\texorpdfstring{\thanks{This work was
partially supported by the DFG (German Research Foundation).}}{}}
\authorrunning{R.J. van Glabbeek, U. Goltz \& J.-W. Schicke-Uffmann}

\author{
  Rob van Glabbeek \inst{1,2} \and
  Ursula Goltz \inst{3} \and
  Jens-Wolfhard Schicke-Uffmann \inst{3}
}
\institute{
  Data61, CSIRO, Sydney, Australia
  \and
  School of Comp.\ Sc.\ and Engineering, Univ.\ of New South Wales,
  Sydney, Australia
  \and
  Institute for Programming and Reactive Systems, TU Braunschweig, Germany
}

\begin{document}

\maketitle
\setcounter{footnote}{0}

\begin{abstract}
Goltz and Reisig generalised Petri's concept of \emph{processes} of one-safe Petri nets
to general nets where places carry multiple tokens.
\emph{BD-processes} are equivalence classes of Goltz-Reisig processes connected through the
swapping transformation of Best and Devillers; they can be considered as an alternative
representation of runs of nets.
Here we present an order respecting bijection between the BD-processes and the \emph{FS-processes}
of a countable net, the latter being defined---in an analogous way---as equivalence classes of
firing sequences. 
Using this, we show that a countable net without binary conflicts has a (unique) largest BD-process.
\end{abstract}

\section{Introduction}

For the basic class of Petri nets, the \emph{condition/event systems}, there is a well
established notion of \emph{process} \cite{petri77nonsequential}, modelling runs of the represented system.
This paper continues the adaptation of this notion of process to general
\emph{place/transition systems} (\emph{P/T systems}).

Goltz and Reisig proposed a notion of process for P/T systems which is rather
discriminating~\cite{goltz83nonsequential}. Depending on which of several ``identical'' tokens you
choose for firing a transition, you may get different processes, with different causal
dependencies. We call this notion a \emph{GR-process}.

\hspace{-2pt}Best and Devillers \cite{best87both} defined a swapping transformation on GR-processes
that identifies GR-processes differing only in the choice which token was removed from a place. 
They proposed an equivalence notion $\equiv_1^\infty$ on GR-processes, where
$\equiv_1^\infty$-equivalent processes intuitively can be converted into each
other through `infinitely many' swapping transformations.
We address an $\equiv_1^\infty$-equivalence class of GR-processes as a \emph{BD-process}.

GR-processes can be seen as an unsatisfactory formalisation of the intuitive concept of a run, since
there exist conflict-free\footnote{Intuitively, a \emph{conflict} denotes any situation in which
  there is a choice to resolve.\\ A formalisation of this notion \cite{goltz86howmany,GGS19} occurs in \refsec{conflict}.}
systems with multiple maximal GR-processes.
We refer to \cite{glabbeek11ipl,GGS19} and the many
references therein for an example and further discussion.
On the other hand,
BD-processes can be seen as unsatisfactory, because there exist systems which do have conflicts, yet
still have a unique maximal BD-process. To illustrate this result, we recall in
\reffig{badswapping} an example due to Ochma\'nski
\cite{ochmanski89personal} --- see also
\cite{DMM89,glabbeek11ipl}. In the initial situation only two of the three
enabled transitions can fire, which constitutes a conflict.  However, the
equivalence $\equiv_1^\infty$ obtained from the swapping transformation (formally defined in
\refsec{semantics}) identifies all possible maximal GR-processes---two of which are shown here---and hence
yields only one complete abstract run of the system.  We are not aware of a
solution, i.e.\ any formalisation of the concept of a run of a net
that allows only one complete run for a conflict-free net
but allocates multiple complete runs to the net of \reffig{badswapping}.

\begin{figure}[t]
\vspace*{1em}
  \begin{center}
    \psscalebox{0.9}{
    \begin{petrinet}(18.95,7)
      \Ql(0.5,5):pa:2;
      \Ql(2,5):pb:3;
      \Q(4,5):pc:4;
      \Q(6.5,5):pd:5;

      \ul(1,4):a::a;
      \ul(2.5,4):b::b;
      \u(4,4):c::c;
      \u(6,4):d::d;

      \ql(2.5,6):p:1;
      \ql(2.5,2):q:6;

      \av p[210]-(1,5)->[90]a; \a pa->a; \a a->q;
      \a p->b; \a pb->b; \a b->q;
      \a p->c; \a pc->c; \a c->q;
      \av q[0]-(6,2)->[270]d; \a pd->d; \av d[90]-(6,6)->[0]p;

      \pscircle*(2.38,6){0.1}
      \pscircle*(2.62,6){0.1}
    \end{petrinet}
    }
    \psline[linestyle=dotted](-1.25,0)(-1.25,5)
    \rput(0.5,3.5){\psscalebox{0.75}{
    \begin{petrinet}(7,12)
      \s(1,8.5):fiveA:5;
      \s(2,8.5):twoA:2;
      \s(3,8.5):oneA:1;
      \s(4,8.5):threeA:3;
      \PNsr(5,8.5):oneB:1:$q$;
      \s(6.5,8.5):fourA:4;

      \t(2.5,7):aA:a;
      \t(4.5,7):bA:b;
      
      \s(2.5,5.5):sixA:6;
      \s(4.5,5.5):sixB:6;

      \t(2,4):dA:d;

      \PNsl(2,2.5):oneC:1:$p$;

      \t(5,1.5):cA:c;

      \s(5,0):sixC:6;

      \a twoA->aA; \a oneA->aA;
      \a threeA->bA; \a oneB->bA;
      \a aA->sixA;
      \a bA->sixB;
      \a sixA->dA; \a fiveA->dA;
      \a dA->oneC;
      \a oneC->cA; \a fourA->cA;
      \a cA->sixC;
    \end{petrinet}
    }}
    \psline[linestyle=dotted](2.5,0)(2.5,5)
    \rput(4.2,3.5){\psscalebox{0.75}{
    \begin{petrinet}(7,12)
      \s(1,8.5):fiveA:5;
      \s(2,8.5):twoA:2;
      \s(3,8.5):oneA:1;
      \s(4,8.5):threeA:3;
      \PNsr(5,8.5):oneB:1:$q$;
      \s(6.5,8.5):fourA:4;

      \t(2.5,7):aA:a;
      \t(4.5,1.5):bA:b;
      
      \s(2.5,5.5):sixA:6;
      \s(5.5,5.5):sixB:6;

      \t(2,4):dA:d;

      \PNsl(2,2.5):oneC:1:$p$;

      \t(5.5,7):cA:c;

      \s(4.5,0):sixC:6;

      \a twoA->aA; \a oneA->aA;
      \a threeA->bA; \a oneC->bA;
      \a aA->sixA;
      \a cA->sixB;
      \a sixA->dA; \a fiveA->dA;
      \a dA->oneC;
      \a oneB->cA; \a fourA->cA;
      \a bA->sixC;
    \end{petrinet}
    }}
  \end{center}
  \vspace{-6ex}
  \begin{center}
  \caption{A net together with two of its maximal GR-processes,}
     \quad which are identified by swapping equivalence.
  \end{center}
  \label{fig-badswapping}\vspace{-5ex}
\end{figure}

In \cite{glabbeek11ipl,GGS19} we propose a
subclass of P/T systems, called \emph{structural conflict nets}, more general then the well-known
class of \emph{safe} nets. On these nets BD-processes are a good formalisation of runs, for
we showed that a structural conflict net has a largest BD-process if and only if the net is conflict-free.

The question remains what happens for general P/T systems.
As we have illustrated above, systems with conflicts may still have one largest BD-process.
In this paper we will show that the ``if'' part of the above-mentioned correspondence also holds
for general countable P/T systems: a countable conflict-free P/T-system has a largest BD-process.%
\footnote{In fact, we present a slightly stronger result, namely that a countable P/T-system
  without \emph{binary} conflicts has a largest BD-process. We also give a counterexample showing
  that this stronger result needs the restriction to \emph{countable} P/T systems.
  We do not know whether each uncountable P/T-system without any conflicts has a largest BD-process.}

However, it turns out that the proof of this result is much more complicated than the special case
for structural conflict nets established in \cite{GGS19}.\footnote{The proof of \cite[Theorem~2]{GGS19},
  creating a largest BD-process for any given structural conflict net, does not generalise
  beyond structural conflict nets. We did not find a better method for this generalisation than via
  the detour of FS-processes, as described below.}

Best and Devillers \cite{best87both} defined a swapping transformation also on the firing sequences of a
net, allowing two adjacent transitions to be swapped if they can be fired concurrently.
They proposed an equivalence notion $\equiv_0^\infty$ on firing sequences, where
$\equiv_0^\infty$-equivalent firing sequences intuitively can be converted into each
other through `infinitely many' swaps.
We address an $\equiv_0^\infty$-equivalence class of firing sequences as an \emph{FS-process}.
Best and Devillers established a bijective correspondence between the BD-processes and the
FS-processes of a countable net.
Here we consider the natural preorders $\sqsubseteq_1^\infty$ on GR-processes and
$\sqsubseteq_0^\infty$ on firing sequences with kernels $\equiv_1^\infty$ and $\equiv_0^\infty$.
They induce partial orders (also denoted $\sqsubseteq_1^\infty$ and $\sqsubseteq_0^\infty$)
on BD-processes and FS-processes respectively. In \refsec{firing sequences} we prove that 
the bijective correspondence between the BD-processes and the
FS-processes of a countable net respects this order, so that a countable net has a largest BD-process iff it
has a largest FS-process. This result is interesting in its own right. Additionally we use it as a
stepping stone for obtaining our main result discussed above, by showing that
a countable conflict-free P/T-system has a largest FS-process.

The results of this paper appeared already in our technical report
\cite{GGS11b}, although formulated and proven differently, since there we
didn't have the preorder $\sqsubseteq_1^\infty$, introduced
in \cite{GGS19}. Our revised proofs are conceptually simpler, as they
avoid the auxiliary concepts of BD-runs and FS-runs.

\section[Place/transition systems]{Place/transition systems\footnotemark}\label{sec-basic}
 \footnotetext{The material in Sections \ref{sec-basic},~\ref{sec-GR} and~\ref{sec-conflict} follows closely
   the presentation in \cite{glabbeek11ipl}, but needs to be included to make the paper self-contained.}

\noindent
We will employ the following notations for multisets.

\begin{define}{
  Let $X$ be a set.
}\label{df-multiset}
\item A {\em multiset} over $X$ is a function $A\!:X \rightarrow \bbbn$,
i.e.\ $A\in \powermultiset{X}\!\!$.
\item $x \in X$ is an \defitem{element of} $A$, notation $x \in A$, iff $A(x) > 0$.
\item For multisets $A$ and $B$ over $X$ we write $A \subseteq B$ iff
 \mbox{$A(x) \leq B(x)$} for all $x \inp X$;
\\ $A\cup B$ denotes the multiset over $X$ with $(A\cup B)(x):=\text{max}(A(x), B(x))$,
\\ $A\cap B$ denotes the multiset over $X$ with $(A\cap B)(x):=\text{min}(A(x), B(x))$,
\\ $A + B$ denotes the multiset over $X$ with $(A + B)(x):=A(x)+B(x)$,
\\ $A - B$ is given by
$(A - B)(x):=A(x)\monus B(x)=\mbox{max}(A(x)-B(x),0)$, and
for $k\inp\bbbn$ the multiset $k\cdot A$ is given by
$(k \cdot A)(x):=k\cdot A(x)$.
\item The function $\emptyset\!:X\rightarrow\bbbn$, given by
  $\emptyset(x):=0$ for all $x \inp X$, is the \emph{empty} multiset over $X$.
\item If $A$ is a multiset over $X$ and $Y\subseteq X$ then
  $A\restrictedto Y$ denotes the multiset over $Y$ defined by
  $(A\restrictedto Y)(x) := A(x)$ for all $x \inp Y$.
\item The cardinality $|A|$ of a multiset $A$ over $X$ is given by
  $|A| := \sum_{x\in X}A(x)$.
\item A multiset $A$ over $X$ is \emph{finite}
  iff $|A|<\infty$, i.e.,
  iff the set $\{x \mid x \inp A\}$ is finite.
\item A function $\pi: X \rightarrow Y$ extends to multisets $A \in \powermultiset{X}$ by
  \plat{$\displaystyle\pi(A)(y) = \!\sum_{y = \pi(x)}\!A(x)$}. In this paper, this sum
  will always turn out to be finite.
\end{define}
Two multisets $A\!:X \rightarrow \bbbn$ and
$B\!:Y\rightarrow \bbbn$
are \emph{extensionally equivalent} iff
$A\restrictedto (X\cap Y) = B\restrictedto (X\cap Y)$,
$A\restrictedto (X\setminus Y) = \emptyset$, and
$B \restrictedto (Y\setminus X) = \emptyset$.
In this paper we often do not distinguish extensionally equivalent
multisets. This enables us, for instance, to use $A \cup B$ even
when $A$ and $B$ have different underlying domains.
With $\{x,x,y\}$ we will denote the multiset over $\{x,y\}$ with
$A(x)\mathbin=2$ and $A(y)\mathbin=1$, rather than the set $\{x,y\}$ itself.
A multiset $A$ with $A(x) \leq 1$ for all $x$ is
identified with the set $\{x \mid A(x)=1\}$.

Below we define place/transition systems as net structures with an initial marking.
In the literature we find slight variations in the definition of \mbox{P\hspace{-1pt}/T}
systems concerning the requirements for pre- and postsets of places
and transitions. In our case, we do allow isolated places. For
transitions we allow empty postsets, but require at least one
preplace, thus avoiding problems with infinite self-concurrency.
Moreover, following \cite{best87both}, we restrict attention
to nets of \defitem{finite synchronisation}, meaning that each
transition has only finitely many pre- and postplaces.
Arc weights are included by defining the flow relation as a function to the natural numbers.
For succinctness, we will refer to our version of a \mbox{P\hspace{-1pt}/T} system as a \defitem{net}.

\begin{define}{}\label{df-nst}
\item[]
  A \defitem{net} is a tuple
  $N = (S, T, F, M_0)$ where
  \begin{itemise}
    \item $S$ and $T$ are disjoint sets (of \defitem{places} and \defitem{transitions}),
    \item $F: ((S \mathord\times T) \mathrel\cup (T \mathord\times S)) \rightarrow \bbbn$
      (the \defitem{flow relation} including \defitem{arc weights}), and
    \item $M_0 : S \rightarrow \bbbn$ (the \defitem{initial marking})
  \end{itemise}
  such that for all $t \inp T$ the set $\{s\mid F(s, t) > 0\}$ is
  finite and non-empty, and the set $\{s\mid F(t, s) > 0\}$ is finite.
\end{define}

\noindent
Graphically, nets are depicted by drawing the places as circles and
the transitions as boxes. For $x,y \inp S\cup T$ there are $F(x,y)$
arrows (\defitem{arcs}) from $x$ to $y$.\footnote{This is a presentational alternative for the
  common approach of having at most one arc from $x$ to $y$, labelled with the \emph{arcweight}
  $F(x,y) \in \bbbn$.}  When a net represents a
concurrent system, a global state of this system is given as a
\defitem{marking}, a multiset of places, depicted by placing $M(s)$
dots (\defitem{tokens}) in each place $s$.  The initial state is
$M_0$.

\begin{define}{
  Let $N\!=\!(S, T, F, M_0)$ be a net and $x\inp S\cup T$.
}\label{df-preset}
\item[]
The multisets $\precond{x},~\postcond{x}: S\cup T \rightarrow
\bbbn$ are given by $\precond{x}(y)=F(y,x)$ and
$\postcond{x}(y)=F(x,y)$ for all $y \inp S \cup T$.
If $x\in T$, the elements of $\precond{x}$ and $\postcond{x}$ are
called \emph{pre-} and \emph{postplaces} of $x$, respectively.
These functions extend to finite multisets
$X{:}\, S \cup T \rightarrow\bbbn$ as usual, by
$\precond{\!X} \mathbin{:=} \sum_{x \in S \cup T}X(x)\cdot\precond{x}$ and
$\postcond{X} \mathbin{:=} \sum_{x \in S \cup T}X(x)\cdot\postcond{x}\!$.
\end{define}
The system behaviour is defined by the possible moves between
markings $M$ and $M'$, which take place when a finite multiset $G$ of
transitions \defitem{fires}.  When firing a transition, tokens on
preplaces are consumed and tokens on postplaces are created, one for
every incoming or outgoing arc of $t$, respectively.  Obviously, a
transition can only fire if all necessary tokens are available in $M$
in the first place. \refdf{firing} formalises this notion of behaviour.

\begin{define}{
  Let $N \mathbin= (S, T, F, M_0)$ be a net,
  $G \in \bbbn^T\!$, $G$ non-empty and finite, and $M, M' \in \bbbn^S\!$.
}\label{df-firing}
\item[]
$G$ is a \defitem{step} from $M$ to $M'$,
written $M\production{G}_N M'$, iff
\begin{itemise}
  \item $^\bullet G \subseteq M$ ($G$ is \defitem{enabled}) and
  \item $M' = (M - \mbox{$^\bullet G$}) + G^\bullet$. 
\end{itemise}
We may leave out the subscript $N$ if clear from context.
For a word $\sigma = t_1t_2\ldots t_n \in T^*$
we write $M\production{\sigma} M'$ for\vspace{-5pt}
$$
\exists M_1, M_2, \ldots, M_{n-1}.
M\!\production{\{t_1\}}\! M_1\!\production{\{t_2\}}\! M_2 \cdots M_{n-1}\!\production{\{t_n\}}\! M'\!\!.
$$
When omitting $\sigma$ or $M'$ we always mean it to be existentially quantified.
Likewise, for an infinite word $\sigma = t_1t_2t_3\ldots \in T^\omega$
we write $M\production{\sigma}$ for\vspace{-5pt}
$$
\exists M_1, M_2, \ldots.
M\!\production{\{t_1\}}\! M_1\!\production{\{t_2\}}\!
M_2\!\production{\{t_3\}}\! \cdots.
$$
When $M_0 \production{\sigma}_N$, the word $\sigma \in T^*\cup T^\omega$
is called a \defitem{firing sequence} of $N$.
The set of all firing sequences of $N$ is denoted by $\wFS$,
and the subset of finite firing sequences of $N$ is denoted by $\fFS$.
\end{define}

\noindent
Note that steps are (finite) multisets, thus allowing self-concurrency.
Also note that $M\goesto[\{t,u\}]$ implies $M\goesto[tu]$ and $M\goesto[ut]$.
We use the notation $t\in \sigma$ to indicate that the transition $t$
occurs in the sequence $\sigma$ and use $\sigma\leq\rho$ to indicate that
$\sigma$ is a prefix of the sequence $\rho$, i.e.\ $\exists \mu. \rho=\sigma\mu$.

\section{Processes of place/transition systems}\label{sec-semantics}

\noindent
We now define two notions of a process of a net, modelling a run of
the represented system on two levels of abstraction.

\subsection{GR-processes}\label{sec-GR}

A (GR-)process is essentially a conflict-free, acyclic net together
with a mapping function to the original net. It can be obtained by
unwinding the original net, choosing one of the alternatives in case
of conflict.
The acyclic nature of the process gives rise to a notion of causality
for transition firings in the original net via the mapping function.
A conflict present in the original net is represented by the existence of
multiple processes, each representing one possible way to decide the conflict.
\pagebreak

\begin{define}{}\label{df-process}
 \item[]
  A pair $\PP = (\NN, \pi)$ is a
  \defitem{(GR-)process} of a net $N = (S, T, F, M_0)$
  iff
  \begin{itemise}\itemsep 3pt
   \item $\NN = (\SS, \TT, \FF, \MM_0)$ is a net, satisfying
   \begin{itemisei}
    \item $\forall s \in \SS. |\precond{s}| \leq\! 1\! \geq |\postcond{s}|
    \wedge\, \MM_0(s) = \left\{\begin{array}{@{}l@{\quad}l@{}}1&\mbox{if $\precond{s}=\emptyset$}\\
                                   0&\mbox{otherwise,}\end{array}\right.$
    \item $\FF$ is acyclic, \ie
      $\forall x \inp \SS \cup \TT. (x, x) \mathbin{\not\in} \FF^+$,
      where $\FF^+$ is the transitive closure of $\{(x,y)\mid \FF(x,y)>0\}$,
    \item and $\{t \in \TT \mid (t,u)\in \FF^+\}$ is finite for all $u\in \TT$.
   \end{itemisei}
    \item $\pi:\SS \cup \TT \rightarrow S \cup T$ is a function with 
    $\pi(\SS) \subseteq S$ and $\pi(\TT) \subseteq T$, satisfying
   \begin{itemisei}
    \item $\pi(\MM_0) = M_0$, i.e.\ $M_0(s) = |\pi^{-1}(s) \cap \MM_0|$ for all $s\in S$, and
    \item $\forall t \in \TT, s \in S.
      F(s, \pi(t)) = |\pi^{-1}(s) \cap \precond{t}| \wedge
      F(\pi(t), s) = |\pi^{-1}(s) \cap \postcond{t}|$, i.e.\
      $\forall t \in \TT. \pi (^\bullet t)= {^\bullet \pi(t)} \wedge
      \pi (t^\bullet)= {\pi(t)^\bullet}$.
\pagebreak[3]
  \end{itemisei}
  \end{itemise}
  $P$ is called \defitem{finite} if $\TT$ is finite. The \defitem{end of} $P$
  is defined as $P^\circ = \{s \in \SS \mid \postcond{s} = \emptyset\}$.%
\end{define}

\noindent
For example \reffig{badswapping} gives a net and two of its GR-processes, in which
each place and transition $x$ is labelled $\pi(x)$.
Let $\GR$ (resp.~$\fGR$) denote the collection of (finite) GR-processes of~$N\!$.

A process is not required to represent a completed run of the original net.
It might just as well stop early. In those cases, some set of transitions can
be added to the process such that another (larger) process is obtained. This
corresponds to the system taking some more steps and gives rise to a natural
order between processes.

\begin{define}{
  Let $\PP = ((\SS, \TT, \FF, \MM_0), \pi)$ and $\PP' = ((\SS', \TT\,', \FF\,', \MM_0'), \pi')$ be
  two processes of the same net.
}\label{df-extension}
\item
  $\PP'$ is a \defitem{prefix} of $\PP$, notation $\PP'\leq \PP$, and 
  $\PP$ an \defitem{extension} of $\PP'$, iff 
    $\SS'\subseteq \SS$,
    $\TT\,'\subseteq \TT$,
    $\MM_0' = \MM_0$,
    $\FF\,'=\FF\restrictedto(\SS' \mathord\times \TT\,' \mathrel\cup \TT\,' \mathord\times \SS')$
    and $\pi'=\pi\restrictedto(\SS'\cup \TT\,')$.
\item
  A process of a net is said to be \defitem{maximal} if 
  it has no proper extension.
\end{define}

\noindent
The requirements above imply that if $\PP'\leq \PP$, $(x,y)\in\FF^+$
and $y\in \SS' \cup \TT\,'$ then $x\in \SS' \cup \TT\,'$.
Conversely, any subset $\TT\,'\subseteq \TT$ satisfying
$(t,u)\in \FF^+ \wedge u\in \TT\,' \Rightarrow t\in \TT\,'$ uniquely determines a
prefix of $\PP$.

In \cite{petri77nonsequential,genrich80dictionary,goltz83nonsequential} processes were
defined without requiring the third condition on $\NN$ from \refdf{process}.
Goltz and Reisig \cite{goltz83nonsequential} observed that certain
processes did not correspond with runs of systems, and proposed to
restrict the notion of a process to those that can be approximated by
finite processes \cite[end of Section~3]{goltz83nonsequential}.
This is the role of the third condition on $\NN$ in \refdf{process}; it
is equivalent to requiring that each transition occurs in a finite prefix.
In \cite{petri77nonsequential,genrich80dictionary,goltz83nonsequential}
only processes of finite nets were considered. For those processes,
the requirement of \emph{discreteness} proposed in \cite{goltz83nonsequential}
is equivalent with imposing the third condition on $\NN$ in \refdf{process}
\cite[Theorem 2.14]{goltz83nonsequential}.

Two processes $\PP \mathbin= (\NN, \pi)$ and $\PP' \mathbin= (\NN\,', \pi')$
are \defitem{isomorphic}, notation $\PP \cong \PP'$, iff there exists
an isomorphism $\phi$ from $\NN$ to $\NN\,'$ which respects the
process mapping, i.e.\ $\pi = \pi' \circ \phi$.
Here an \emph{isomorphism} $\phi$ between two nets $\NN=(\SS, \TT, \FF,
\MM_0)$ and $\NN\,'=(\SS', \TT\,', \FF\,', \MM'_0)$ is a
bijection between their places and transitions such that
$\MM'_0(\phi(s))=\MM_0(s)$ for all $s\in\SS$ and
$\FF\,'(\phi(x),\phi(y))=\FF(x,y)$ for all $x,y\in \SS\cup\TT$.

\subsection{BD-processes}

Next we formally introduce the swapping transformation and the resulting
equivalence notion on GR-processes from \cite{best87both}.

\begin{define}{
  Let $\PP = ((\SS, \TT, \FF, \MM_0), \pi)$ be a process and
  let $p, q \in \SS$ with $(p,q) \notin \FF^+\cup (\FF^+)^{-1}$ and
  $\pi(p) = \pi(q)$.
  }
\label{df-swap}
\item[]
  Then $\swap(\PP, p, q)$ is defined as $((\SS, \TT, \FF\,', \MM_0), \pi)$ with
  \begin{equation*}
    \FF\,'(x, y) = \begin{cases}
      \FF(q, y) & \text{ iff } x = p,\, y \in \TT\\
      \FF(p, y) & \text{ iff } x = q,\, y \in \TT\\
      \FF(x, y) & \text{ otherwise. }
    \end{cases}
  \end{equation*}
\end{define}

\noindent
We refer to \cite{best87both,GGS19} for an explanation of this definition and further examples.
Here we only give the processes of \reffig{badswapping} as being connected via $\swap$.

\begin{define}{}\label{df-swapeq}
\item
  Two processes $\PP$ and $\QQ$ of the same net are
  \defitem{one step swapping equivalent} ($\PP \swapeq \QQ$) iff
  $\swap(\PP, p, q)$ is isomorphic to $\QQ$ for some places $p$ and $q$.
\item
We write $\swapeq^*$ for the reflexive and transitive closure of $\swapeq$.
\end{define}
In \cite[Definition 7.8]{best87both} swapping equivalence---denoted $\equiv_1^\infty$---is defined
in terms of \emph{reachable B-cuts}. In \cite{GGS19} this definition was reformulated as
follows, also introducing the associated preorder $\sqsubseteq_1^\infty$.

\begin{define}{Let $N$ be a net, and $P,Q\in \GR$.}\label{df-BD-swapping-alt}
\item[]
Then $\PP \sqsubseteq_1^\infty \QQ$ iff\vspace{-1ex}
\begin{equation*}
\forall P''\inp\fGR, P'' \leq P. \exists P',Q'\in\fGR.
  P'' \leq P' \swapeqi^* Q' \leq Q.
\end{equation*}
Moreover, $\PP \equiv_1^\infty \QQ$ iff $\PP \sqsubseteq_1^\infty \QQ \wedge \QQ \sqsubseteq_1^\infty \PP$.
\end{define}
Thus, $\PP \sqsubseteq_1^\infty \QQ$ holds if and only if each finite prefix of $\PP$ 
can be extended into a finite process that is $\swapeqi^*$-equivalent to a prefix of $\QQ$.

In \cite{GGS19} it is shown that $\sqsubseteq_1^\infty$ is a preorder, and thus $\equiv_1^\infty$ an
equivalence relation on GR-processes. Trivially, $\swapeq^*$ is included in $\swapeq^\infty$.

\begin{define}{}
\item[]
We call a $\swapeq^\infty$-equivalence class of GR-processes
a \defitem{BD-process}.
\end{define}

\section{Conflicts in place/transition systems}\label{sec-conflict}

We recall the canonical notion of conflict introduced in
\cite{goltz86howmany}.

\begin{define}{
  Let $N \mathbin= (S, T, F, M_0)$ be a net and $M \in
  \bbbn^S\!$.
}\label{df-semanticconflict}
\item
  A finite, non-empty multiset $G \in \bbbn^T$ is in
  \defitem{(semantic) conflict} in $M$ iff\vspace{2pt}\\
  $\neg M\goesto[G] ~~\wedge~~ \forall t \in G. M\goesto[G \restrictedto \{t\}]$.
\item
  $N$ is \defitem{(semantic) conflict-free} iff
  no finite, non-empty multiset $G \in \bbbn^T$ is in semantic conflict in any
  $M$ with $M_0 \goesto[] M$.
\item
  $N$ is \defitem{binary-conflict-\!-free} iff
  no multiset $G \in \bbbn^T$ with $|G| = 2$ is in semantic conflict in any
  $M$ with $M_0 \goesto[] M$.
\end{define}
Thus, $N$ is binary-conflict-\!-free iff whenever two different
transitions $t$ and $u$ are enabled at a reachable marking $M$, then
also the step $\{t,u\}$ is enabled at $M$.
The above concept of (semantic) conflict-freeness formalises the intuitive notion that
there are no choices to resolve.
In \cite{GGS19} the above definition is compared with other notions of conflict and
conflict-freeness that occur in the literature.

A finite multiset $G$ of transitions has a \emph{structural conflict}
iff it contains two different transitions that share a preplace.
We proposed in \cite{glabbeek11ipl} a class of \mbox{P\hspace{-1pt}/T} systems
where this structural definition of conflict matches the semantic definition of conflict as
given above. We called this class of nets \defitem{structural conflict
nets}\footnote{This class pertains only to the context of this work and reappears in the conclusion.}.
For a net to be a structural conflict net, we require that two
transitions sharing a preplace will never occur both in one step.

\begin{define}{
  Let $N \mathbin= (S, T, F, M_0)$ be a net.
}\label{df-structuralconflict}
\item[]
  $N$ is a \defitem{structural conflict net} iff
  $\forall t, u.
    (M_0 \goesto[]\;\goesto[\{t, u\}]) \implies
    \precond{t} \cap \precond{u} = \emptyset$.
\end{define}
Note that this excludes self-concurrency from the possible behaviours in a
structural conflict net: as in our setting every transition has at least one
preplace, $t = u$ implies $\precond{t} \cap \precond{u} \ne \emptyset$.
Also note that in a structural conflict net a non-empty, finite multiset $G$ is in
conflict in a reachable marking $M$ iff $G$ is a set, each transition
from $G$ is  enabled at $M$, and and two distinct transitions in $G$
are in conflict in $M$. Hence a structural conflict net is conflict-free if
and only if it is binary-conflict-\!-free.  Moreover, two transitions enabled
in $M$ are in (semantic) conflict iff they share a preplace.

\vspace{-2ex}
\section{Characterising BD-processes by firing sequences}
\label{sec-firing sequences}
\vspace{-0.5ex}

In Section~\ref{sec-semantics} a BD-process was defined as a
$\swapeq^\infty$-equivalence class of GR-processes; moreover the
preorder $\sqsubseteq_1^\infty$ on GR-processes induces a partial
order on BD-processes, and hence a concept of a largest BD-process.

Best and Devillers \cite{best87both} introduced an equivalence relation
$\equiv_0^\infty$ on the firing sequences of a countable net, such that the
BD-processes are in a bijective correspondence with the 
$\equiv_0^\infty$-equivalence classes of firing sequences, called
\emph{FS-processes} in \cite{glabbeek11ipl}.
In this section we define a preorder $\sqsubseteq_0^\infty$ on the
firing sequences of a net, with kernel $\equiv_0^\infty$, that thus induces
a partial order on FS-processes, and hence a concept of a largest FS-process.\pagebreak
We show that the bijection between BD-processes and FS-processes respects these orders, and
therefore also the associated notion of a largest process.
Thus a countable net has a largest BD-process iff it has a largest FS-process.

Our main result, that a countable P/T system without binary-conflict
has a largest BD-process, can therefore be established in terms of FS-processes.

\subsection{FS-processes}

The behaviour of a net can be described not only by its processes, but also by
its firing sequences. The imposed total order on transition firings abstracts
from information on causal dependence, or concurrency, between transition
firings.  To retrieve this information we introduce an \defitem{adjacency}
relation on firing sequences, recording which interchanges of transition
occurrences are due to semantic independence of transitions. Hence adjacent
firing sequences represent the same run of the net. We then define
\defitem{FS-processes} in terms of the resulting equivalence classes of firing
sequences. Adjacency is similar to the idea of Mazurkiewicz traces
\cite{mazurkiewicz95tracetheory}, allowing to exchange concurrent transitions.
However, it is based on the semantic notion of concurrency instead of
the global syntactic independence relation of trace theory, similar as
in the approach of generalising trace theory in \cite{HKT95}.
Further discussion on adjacency can be found in
\cite{carstensen91,vogler90swapping}. Carstensen \cite{carstensen91} studies the
complexity of the relation $\connectedto$ defined below;
Vogler \cite{vogler90swapping} finds canonical representatives of
$\connectedto$-equivalence classes for a restricted class of nets.

\begin{define}{
  Let $N = (S, T, F, M_0)$ be a net, and $\sigma, \rho \in \FSinf(N)$.
}\label{df-connectedto}
\item
  $\sigma$ and $\rho$ are \defitem{adjacent}, $\sigma \adjacent \rho$,
  iff $\sigma = \sigma_1 t u \sigma_2$, $\rho = \sigma_1 u t \sigma_2$ and
  $M_0 \goesto[\sigma_1]\goesto[\{t, u\}]$.
\item
  We write $\connectedto$ for the reflexive and transitive closure of $\adjacent$.
\end{define}

\noindent
Note that $\connectedto$-related firing sequences contain the same
multiset of transition occurrences.
When writing $\sigma \connectedto \rho$ we implicitly claim that
$\sigma, \rho \in \FSinf(N)$.
Furthermore $\sigma \connectedto \rho \wedge \sigma\mu \in \FSinf(N)$
implies $\sigma \mu \connectedto \rho \mu$ for all $\mu \in T^* \cup T^\omega$.

\begin{lemma}\rm\label{lem-swapprefix-fs1}
Let $N = (S, T, F, M_0)$ be a net, let $\sigma_1 \adjacent \sigma_2 \leq \sigma_3$ for some
$\sigma_1,\sigma_2 \inp \FS(N)$ and $\sigma_3 \inp \FSinf(N)$.
Then there is a
$\sigma'\inp \FSinf(N)$ with $\sigma_1 \leq \sigma' \adjacent \sigma_3$.
Moreover, if $\sigma_3 \in \FS(N)$ then $\sigma' \in \FS(N)$.
\end{lemma}
\begin{proof}
  We have that $\sigma_1 = \alpha tu\beta$,
  $\sigma_2 = \alpha ut \beta$,
  $M_0 \goesto[\alpha]\goesto[\{t, u\}] M_1$ and
  $\sigma_3 = \sigma_2 \gamma$ for some
  $\alpha, \beta \in T^*$, $\gamma \in T^* \cup T^\omega$ and $M_1 \in \bbbn^S$.
  Naturally then, we take $\sigma' = \alpha tu \beta\gamma$.
  From $M_0 \goesto[\alpha] \goesto[\{t, u\}] M_1$ follows
  $M_0 \goesto[\alpha tu] M_1$
  and $M_0 \goesto[\alpha ut] M_1$.
  From $\sigma_3 \in \FSinf(N)$ follows additionally
  $M_1 \goesto[\beta\gamma]$.
  Hence $M_0 \goesto[\alpha tu] M_1 \goesto[\beta\gamma]$ and
  $\sigma' \in \FSinf(N)$. The case $\sigma_3 \inp \FS(N)$ follows
  similarly.
  That  $\sigma_1 \leq \sigma' \adjacent \sigma_3$ holds trivially.
  \qed
\end{proof}

\begin{corollary}\rm\label{cor-swapprefix-fs}
Let $\sigma_1 \connectedto \sigma_2 \leq \sigma_3$ for some
$\sigma_1,\sigma_2 \inp \FS(N)$ and $\sigma_3 \inp \FSinf(N)$.
Then there is a
$\sigma'\inp \FSinf(N)$ with $\sigma_1 \leq \sigma' \connectedto \sigma_3$.
Moreover, if $\sigma_3 \in \FS(N)$ then $\sigma' \in \FS(N)$.
\qed
\end{corollary}
\begin{lemma}\rm\label{lem-finite-swap}
  Let $\sigma''\mathbin\in\FS(N)$, $\rho\mathbin\in\FSinf(N)$. Then
  $\exists \rho^\dagger\in\FSinf(N). \sigma''\leq \rho^\dagger\connectedto \rho$ iff
$\exists \sigma',\rho'\in\FS(N).   \sigma'' \leq \sigma' \equiv_0^* \rho' \leq \rho$.
\end{lemma}
\begin{proof}
``If'' follows by \refcor{swapprefix-fs}.
For ``only if'' take $\rho'$ to be the smallest prefix of $\rho$ that contains all transitions
interchanged between $\rho$ and $\rho^\dagger$.
\qed
\end{proof}

\noindent
For firing sequences $\sigma,\rho\in\FSinf(N)$,
$\sigma \connectedto \rho$ means that $\sigma$ can be transformed into $\rho$ by repeated
exchange of two successive transitions that can fire concurrently. 
However, $\connectedto$ allows for only finitely many permutations.
In \cite{best87both} a relation $\equiv^\infty_0$ on $\FSinf(N)$ is defined that in some
sense allows infinitely many permutations:

\begin{defineCite}{\cite{best87both}}{Let $N$ be a net, and $\sigma,\rho\inp \FSinf(N)$.}
\item
Write $\sigma \n \rho$ when $\sigma$ and $\rho$ are equal or
both have the same length $\geq n$ (possibly infinite) and agree on the prefix of length $n$.
\item
Then $\sigma \equiv^\infty_0 \rho$ iff
$\forall n\in \bbbn. \exists \sigma',\rho'\in \FSinf(N).
(\sigma\connectedto\sigma' \n \rho \wedge
 \sigma \n \rho' \connectedto \rho)$.
\end{defineCite}
\begin{observation}\rm
$\sigma \equiv^\infty_0 \rho$ iff
$\forall \rho''\inp\FS(N), \rho'' \leq \rho. \exists \sigma'\in\FSinf(N). \rho''\leq \sigma'\connectedto \sigma$
and
$\forall \sigma''\inp\FS(N), \sigma'' \leq \sigma. \exists \rho'\in\FSinf(N). \sigma''\leq \rho'\connectedto \rho$.
\end{observation}
In words, $\sigma \equiv^\infty_0 \rho$ holds iff each finite prefix $\sigma''$ of $\sigma$ is also a prefix of some firing sequence $\rho'$ that is $\connectedto$-equivalent to $\rho$, and vice versa.

Analogously, $\sqsubseteq^\infty_0$ should be the binary relation on
$\FSinf(N)$ given by $\sigma \sqsubseteq^\infty_0 \rho$ iff
$\forall \sigma''\inp\FS(N), \sigma'' \leq \sigma. \exists \rho'\in\FSinf(N). \sigma''\leq \rho'\connectedto \rho$.
By \reflem{finite-swap} the  $\connectedto$-conversion of $\rho$ into $\rho'$ can be
  done in a finite prefix of $\rho$.
  % {\reflem{finite-swap}}
  This allows us to state the formal definition of $\sqsubseteq_0^\infty$ as follows, which will be advantageous
later on:
\begin{define}{Let $N$ be a net, and $\sigma,\rho\in \wFS$.}\label{df-BD-swapping-fs-alt}
\item[]
Then $\sigma \sqsubseteq_0^\infty \rho$ iff\vspace{-1ex}
$$\forall \sigma''\inp\fFS, \sigma'' \leq \sigma. \exists \sigma',\rho'\in\fFS.
  \sigma'' \leq \sigma' \equiv_0^* \rho' \leq \rho.$$
\end{define}
\begin{observation}\rm
$\sigma \equiv_0^\infty \rho$ iff 
$\sigma \sqsubseteq_0^\infty \rho$ and $\rho \sqsubseteq_0^\infty \sigma$.
\end{observation}
\begin{proposition}\rm\label{pr-swapprefix-fs-preorder}
$\sqsubseteq_0^\infty$ is a preorder on $\FSinf(N)$. So $\equiv_0^\infty$
  is an equivalence relation.
\end{proposition}
\begin{proof}
By definition $\sqsubseteq_0^\infty$ is reflexive.
Moreover, \refcor{swapprefix-fs}, in combination with transitivity of $\leq$ and
$\connectedto$, implies transitivity of $\sqsubseteq_0^\infty$:
  Suppose $\sigma \sqsubseteq_0^\infty \rho \sqsubseteq_0^\infty \nu$.
  To obtain $\sigma \sqsubseteq_0^\infty \nu$, let $\sigma'$ be a finite prefix of $\sigma$.
  We need to find a finite prefix $\nu'$ of $\nu$ with  $\sigma' \leq \connectedto \nu'$.
  Since $\sigma \sqsubseteq_0^\infty \rho$, there is a finite prefix $\rho'$ of $\rho$ such that
  $\sigma' \leq \connectedto \rho'$.
  Since $\rho \sqsubseteq_0^\infty \nu$, there is a finite prefix $\nu'$ of $\nu$ such that
  $\rho' \leq \connectedto \nu'$.
  So  $\sigma' \leq \connectedto\leq \connectedto \nu'$, and by \refcor{swapprefix-fs}
  we obtain
  $\sigma' \leq \leq \connectedto \connectedto \nu'$.
\qed
\end{proof}

\noindent
Now an \emph{FS-process} of a net $N$ can be defined as an
$\equiv^\infty_0$-equivalence class of possibly infinite firing
sequences of $N$ (elements of $\FSinf(N)$). Since
$\equiv^\infty_0$ is the kernel of $\sqsubseteq^\infty_0$, the latter
introduces a partial order on FS-processes, and hence a notion of a
largest FS-process.

\begin{figure}[t]
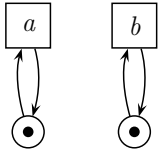

  \begin{center}
    \begin{petrinet}(5,3)
      \P(1,0):s;
      \P(3,0):p;

      \t(1,2):a:a;
      \t(3,2):b:b;
      
      \A s->a; \A a->s;
      \A p->b; \A b->p;
    \end{petrinet}
  \end{center}
  \vspace{-2ex}
  \caption{A net with two self-loops.}
  \label{fig-infiniswap}
\end{figure}

As an example, consider the net of \reffig{infiniswap} and the two infinite firing sequences
$\sigma := abababab\ldots$ and $\rho := abbabbabbabb\ldots$. Since infinitely many transitions
would need to be exchanged, $\sigma \not\connectedto \rho$. However $\sigma \equiv^\infty_0 \rho$.
To explain the direction $\sigma \sqsubseteq^\infty_0 \rho$, for any prefix $\sigma'$ of $\sigma$
with length $n$, a prefix $\rho'$ of $\rho$ with at least length
$\frac{3}{2}n$ has enough $a$s and $b$s to rearrange it such that $\sigma'$ becomes a prefix.
Swapping $n$ elements to the correct place from within a $\frac{3}{2}n$ long $\rho$ will
be possible in no more than $\frac{3}{2}n^2$ swaps.

\subsection{A bijection between FS-processes and countable BD-processes}

We now recapitulate
a result from \cite{best87both},
establishing a correspondence between the GR-processes of a countable net
and its firing sequences.

\begin{definition}[\cite{best87both}]\rm\label{df-compatible}
  Let $N$ be a net, $P = ((\SS, \TT, \FF, \MM_0), \pi) \in \GR$ and $\sigma = t_0t_1t_2\ldots \in \FSinf(N)$.
  If $\sigma$ is finite, let $I = \{i \mid \bbbn \ni i < |\sigma|\}$; otherwise let $I = \bbbn$.

  $P$ and $\sigma$ are \defitem{compatible} iff
  there is a bijection $\pos: \TT \into I$ such that
  \begin{enumerate}
    \item $\forall t \in \TT. \pi(t) = t_{\pos(t)}$
    \item $\forall t, t' \in \TT. (t, t') \in \FF^+ \implies \pos(t) < \pos(t')$.
  \end{enumerate}
\end{definition}

\begin{definition}[\cite{best87both}]\rm\label{df-lin}
  Let $N$ be a net and $P \in \GR$. 

  $\Lin(P) := \{\sigma \mid \sigma \in \FSinf(N) \text{ and $\sigma$ is compatible with } P\}$.
\end{definition}

\begin{theorem}[\cite{best87both}]\rm\label{thm-BD-infty}
Let $\sigma,\rho\mathbin\in\FS^\infty(N)$ and $P,Q \mathbin\in \GR$ such that
$\sigma \mathbin\in \Lin(P)$ and $\rho\mathbin\in\Lin(Q)$.
Then $\sigma \equiv_0^\infty \rho$ iff $P \equiv_1^\infty Q$.
\end{theorem}

\begin{definition}\rm\label{df-countable}
A net, or a GR-process, is \emph{countable} iff it has countably many transitions.
A BD-process is countable iff it is an equivalence class of countable GR-processes.
\end{definition}
The last definition uses that if $P \equiv_1^\infty Q$, then $P$ is countable iff $Q$ is countable.
By definition a finite net is countable. Since each transition in our nets has only
finitely many pre- and postplaces, a countable net has only countably many arcs, and countably many
places, at least when not counting isolated places, which are irrelevant.

\begin{proposition}[{\cite[Construction 3.9 and Theorem 3.13]{best87both}}]\rm\label{pr-countable}
Let $N$ be a net.

For each firing sequence $\sigma\mathbin\in\FS^\infty(N)$ there exists a process $P\mathbin\in\GR$ such
that $\sigma\mathbin\in\Lin(P)$.
Moreover, for each countable process $P\mathbin\in\GR$ there exists a firing sequence
$\sigma\mathbin\in\FS^{\infty}(N)$ such that $\sigma\mathbin\in\Lin(P)$.
\end{proposition}
Obviously, for an uncountable process $P\mathbin\in\GR$ there exists no firing sequence
$\sigma\mathbin\in\Lin(P)$.
In \cite{best87both} only countable nets are considered, and there \refthm{BD-infty}, together with
\refpr{countable}, establishes a bijection between $\equiv_0^\infty$-equivalence classes of firing
sequences and $\equiv_1^\infty$-equivalence classes of GR-processes, or, in our terminology, between
FS-processes and BD-processes. When allowing uncountable nets, we obtain a bijection
between FS-processes and countable BD-processes.

The following theorem says that this bijection preserves the order between FS- and BD-processes.
\begin{theorem}\rm\label{thm-ordering}
Let $\sigma,\rho\in\FSinf(N)$ and $P,Q \in \GR$ such that
$\sigma \in \Lin(P)$ and $\rho\in\Lin(Q)$.
Then $\sigma \sqsubseteq_0^\infty \rho$ iff $P \sqsubseteq_1^\infty Q$.
\end{theorem}
Together with \refpr{countable} this theorem establishes an order-preserving bijection between the
FS-processes and the countable BD-processes of a net. Consequently, a countable net has a largest
BD-process iff it has a largest FS-process.

Although the preorders $\sqsubseteq_0^\infty$ and $\sqsubseteq_1^\infty$ are not considered in \cite{best87both},
the proof of \refthm{BD-infty} in \cite{best87both} can be adapted in a fairly straightforward way to yield
a proof of \refthm{ordering} as well. A more detailed proof of \refthm{ordering}, and
thereby also of \refthm{BD-infty}, using the notation of the present paper, is presented below.

\subsection{This bijection preserves the order between processes}

The next three lemmas say that if a process $P$ is compatible with a firing sequence $\sigma$,
then\vspace{-1ex}
\begin{itemize}
\item any finite extension of $\sigma$ can be matched by a compatible extension of $P$,
\item any finite prefix of $\sigma$ can be matched by a compatible prefix of $P$, and
\item any finite extension of $P$ can be matched by a compatible extension of $\sigma$.
\end{itemize}

\begin{lemma}\rm\label{lem-prefix-up-fs}
Let $P'' \mathbin\in \fGR$, $\sigma'' \mathbin\in \Lin(P'')$ and
$\sigma'\in\fFS$ with $\sigma''\leq \sigma'$.
Then there is $P'\in\fGR$ with $\sigma'\mathbin\in\Lin(P')$ and $P''\leq P'$.
\end{lemma}
\begin{proof}
  We have that $\sigma'' \alpha = \sigma'$ for some $\alpha \in T^*$.
  Using induction over the length of $\alpha$
  we need to prove the claim only for $\sigma'' t = \sigma'$ for arbitrary $t \in T$.
  Let $P''=((\TT'', \SS'', \FF'', \MM_0''), \pi'')$.
  From $\sigma' \in \FS(N)$ follows $M_0 \goesto[\sigma''] M_1 \goesto[t]$.
  From $\sigma'' \in \Lin(P'')$ we have that $M_1 = \pi''(P''^\circ)$
  via Construction~3.9, Proposition~3.10, and Theorem~3.13 of \cite{best87both}.
  Hence $\precond{t} \subseteq M_1 = \pi''(P''^\circ)$.
  From $P''^\circ$ we select a set of \hbox{(pre-)}places $A$ with $\pi''(A) = \precond{t}$ and
  we create a set of fresh (post-)places $B$ together with a function $\pi_B: B \into S$ such that
  $\pi_B(B) = \postcond{t}$.

  We define $P'$ as $((\TT', \SS', \FF', \MM_0'), \pi') := ((\TT'' \cup \{t'\}, \SS'' \cup B,
  \FF'' \cup \{(a, t', 1) \mid a \in A\} \cup \{(t', b, 1) \mid b \in B\},
  \MM_0''), \pi'' \cup \{(t', t)\} \cup \pi_B)$.
  We need to show that $P'$ is a process of $N$, $P'' \leq P'$ and $\sigma' \in \Lin(P')$.

  ``$P'$ is a process of $N$'':
  $\forall s \in \SS'. |\precond{s}| \leq 1$ as the only new entries in $\FF'$ which lead to places are for
  the new places from $B$, where it holds.
  $\forall s \in \SS'. 1 \geq |\postcond{s}|$ as the only new entries in $\FF'$ which lead from places are for
  places $s$ from $P''^\circ$ for which $|\postcond{s}| = 0$ in $P''$.
  Additionally $\MM_0'' = \MM_0'$.
  $\FF'$ is acyclic as the new entries all contain $t'$ and $B$ is disjoint from $\SS''$.
  $\{t \in \TT' \mid (t, u) \in \FF'^+\}$ is finite for all $u \in \TT'$
  since $\TT''$ and hence $\TT'$ is finite.
  $\pi'(\MM_0') = \pi''(\MM_0'') = M_0$ as $B$ is distinct from $\SS''$, hence $\pi_B$ contributes nothing to $\pi'(\MM_0')$.
  Finally
  $\pi'(\precond{t'}) = \pi'(A) = \pi''(A) = \precond{t} = \precond{\pi'(t')}$
  and
  $\pi'(\postcond{t'}) = \pi'(B) = \pi_B(B) = \postcond{t} = \postcond{\pi'(t')}$.
  Hence $P'$ is indeed a process.

  ``$P'' \leq P'$'': As $P'$ was constructed from $P''$ using disjoint unions, this follows immediately.

  ``$\sigma' \in \Lin(P')$'': From $\sigma'' \in \Lin(P'')$ we get a $\pos''$ function. We define
  $\pos' := \pos'' \cup (t', |\sigma'| - 1)$. Checking
  \refdf{compatible} we find $\pi'(t') = t$, which is the last transition in
  $\sigma'$.
  As $|\sigma'| - 1$ is one larger than the largest value of $\pos''$, $\pos'$ is a bijection, and since there is no
  $u' \in \TT'$ with $(t', u') \in \FF'^+$ we conclude that $\sigma'$ is compatible with $P'$.
  \qed
\end{proof}

\begin{lemma}\rm\label{lem-prefix-down-fs}
Let $P \mathbin\in \GR$, $\sigma \in \Lin(P)$ and
$\sigma''\in\fFS$ with $\sigma''\leq \sigma$.
Then there is $P''\in\fGR$ with $\sigma''\mathbin\in\Lin(P'')$ and $P''\leq P$.
\end{lemma}
\begin{proof}
  To be precise,
  let, in this proof only,
  ${}^\bullet x$, $x^\bullet$ denote pre- respectively post-sets in $N$,
  ${}^\circ x$, $x^\circ$ denote pre- respectively post-sets in $P$, and
  ${}^\star x$, $x^\star$ denote pre- respectively post-sets in $P''$.

  We have that $\sigma'' \alpha = \sigma$ for some $\alpha \in T^* \cup T^\omega$.

  Let $P = ((\SS, \TT, \FF, \MM_0), \pi)$.
  From $\sigma \in \Lin(P)$ we get a bijection $\pos$ between $\TT$ and the elements of $\sigma$.
  As $\alpha \subseteq \sigma$, we can take the preimage $\phi := \pos^{-1}(\alpha)$.
  We define $\SS'' := \SS \setminus \phi^\circ$ and $\TT'' := \TT \setminus \phi$ and
  take $$P'' = ((\SS'', \TT'', \FF \restrictedto (\SS'' \times \TT''
  \cup \TT'' \times \SS''), \MM_0), \pi \restrictedto (\SS'' \cup \TT'')).$$
  We need to show that $P''$ is a finite process of $N$, $P'' \leq P$ and $\sigma'' \in \Lin(P'')$.

  ``$P''$ is a process of $N$'':
  As elements were only removed from $P$ and none of them were from $\MM_0$, all clauses of \refdf{process} but the last follow.
  It remains to be shown that for all $t \in \TT''$ we have $\pi({}^\star t) = \precond{\pi(t)} \wedge
  \pi(t^\star) = \postcond{\pi(t)}$. By processhood of $P$ we already have $\pi({}^\circ t) = \precond{\pi(t)} \wedge
  \pi(t^\circ) = \postcond{\pi(t)}$.

  ``$\pi({}^\star t) = \precond{\pi(t)}$'': By ${}^\star t = {}^\circ t$, as follows:
  Take any $s \in {}^\circ t$. If we had $s \in t'^\circ$ for any $t' \in \phi$, then
  $\pos(t) < |\sigma''|$ and $\pos(t') \geq |\sigma''|$ (from their order in $\sigma$)
  but also $(t', t) \in \FF^+$ and thus $\pos(t') < \pos(t)$ (from compatibility of $P$ and $\sigma$).
  Hence we cannot have such a $t'$. Thus $s \not\in \phi^\circ$, $s \in \SS''$ and $s \in {}^\star t$.

``$\pi(t^\star) = \postcond{\pi(t)}$'': By $t^\star = t^\circ$, as follows:
  Take any $s \in t^\circ$. As $t \not\in\phi$ and $|{}^\circ s| \leq 1$ for all $s \in \SS$, we have $s \not\in \phi^\circ$,
  $s \in \SS''$ and $s \in t^\star$.

  ``$P$ is finite'': This follows since $\TT''=\pos^{-1}(\sigma'')$ and $\sigma''$ is finite.

  ``$P'' \leq P$'': This follows immediately from the construction of $P''$.

  ``$\sigma'' \in \Lin(P'')$'':
  Using $\text{pos} \restrictedto \TT''$ it follows that $\sigma''$ is compatible with $P''$.
  \qed
\end{proof}

\begin{lemma}\rm\label{lem-prefix-up-GR}
Let $P'',P' \mathbin\in \fGR$ with $P''\leq P'$, and let $\sigma''
\mathbin\in \Lin(P'')$.
Then there is a $\sigma_0\mathbin\in\Lin(P')$ with $\sigma'' \leq \sigma_0$.
\end{lemma}
\begin{proof}
  Let $P' = ((\SS', \TT', \FF', \MM_0'), \pi')$ and $P'' = ((\SS'', \TT'', \FF'', \MM_0''), \pi'')$.
  Applying induction over the number of transitions in $\TT'$, we can restrict attention
  to the case where $\TT' = \TT'' \cup \{t'\}$ for one new transition $t'$.

  We take $\sigma_0 = \sigma'' \pi'(t')$ and need to show that $\sigma_0 \mathbin\in\Lin(P')$ (for
  by construction $\sigma'' \leq \sigma_0$).
  As $\sigma'' \mathbin\in \Lin(P'')$, it is compatible with $P''$, so there exists
  a bijection $\text{pos}'': \TT'' \into \{0, \ldots, |\sigma''| - 1\}$ as per \refdf{compatible}.

  We define $\text{pos}_0: \TT' \into \{0, \ldots, |\sigma_0| - 1\}$ as
  $\text{pos}_0(t) := \text{pos}''(t)$ iff $t \ne t'$ and
  $\text{pos}_0(t') := |\sigma_0| - 1 = |\sigma''|$ and need to show that $\sigma_0$ is
  \hyperref[df-compatible]{compatible} with $\PP'$:
  \begin{enumerate}
    \item For all $t''\in\TT''$ we have $\pi'(t'')=\pi''(t'')=t_{{\rm pos}''(t'')}=t_{{\rm pos}_0(t'')}$.
      Furthermore, by construction $t_{\text{pos}_0(t')} = t_{|\sigma_0| - 1} = \pi'(t')$.
    \item For all $u, u' \in \TT'$ with $(u, u') \in {\FF'}^+$ we need to show $\text{pos}_0(u) < \text{pos}_0(u')$.
      If $u \ne t' \ne u'$ then this follows from the compatibility of $\text{pos}''$.
      For $u = t'$ there cannot be any $(u, u') \in {\FF'}^+$ because $t'$ was added last in an extension to a process.
      If $u' = t'$ we find that by definition $(t', t') \not\in {\FF'}^+$, and for all other possible $u$,
      $\text{pos}_0(u) = \text{pos}''(u) \leq |\sigma''| - 1 < |\sigma''| = \text{pos}_0(t')$.
  \end{enumerate}
  Finally, we show that $\sigma_0 \mathbin\in \FS(N)$.\vspace{2pt}
  Since $P'' \leq P'$ we have $\precond{t'} \subseteq P''^\circ$.
  So $\precond{\pi'(t')} = \pi'(\precond{t'}) \subseteq \pi'(P''^\circ) = \pi''(P''^\circ)$.
  Moreover, as $\sigma'' \mathbin\in \Lin(P'')$ we have \plat{$M_0 \production{\sigma''} M$},
  where $M=\pi''(P''^\circ)$ via Construction~3.9,\vspace{1pt} Proposition~3.10, and Theorem~3.13
  of \cite{best87both}. Hence $M\goesto[\pi'(t')]$ and $\sigma_0 \mathbin\in \FS(N)$.
  \qed
\end{proof}
In line with the last three lemmas, one might expect that
if a process $P$ is compatible with a firing sequence $\sigma$,
then any finite prefix of $P$ can be matched by a compatible prefix of $\sigma$.
This, however, is obviously false. 
Take for instance a process $P$ consisting of two parallel transitions $a$ and $b$, with the
compatible firing sequence $ab$; now the prefix of $P$ containing only the transition $b$ has no
compatible prefix of $ab$.
The following is our best approximation.

\begin{lemma}\rm\label{lem-prefix-down-GR}
Let $P'' \mathbin\in \fGR$ and $P \mathbin\in \GR$ with
$P''\leq P$, and let $\sigma \mathbin\in \Lin(P)$.
Then there are $\sigma''\mathbin\in\Lin(P'')$ and
$\sigma_1,\sigma_2\in\fFS$ with
$\sigma'' \leq \sigma_1 \equiv_0^* \sigma_2 \leq \sigma$.
\end{lemma}

\begin{proof}
  Let $P = ((\SS, \TT, \FF, \MM_0), \pi)$ and $P'' = ((\SS'', \TT'', \FF'', \MM_0), \pi'')$.
  From $\sigma \in \Lin(P)$ we have a bijection $\text{pos}$ between $\TT$ and the indices of $\sigma$.
  Every finite process can be linearised to a firing sequence. Hence there exists some $\sigma'' \in \Lin(P'')$.
  Thence we obtain a bijection $\text{pos}'': \TT'' \into \{0 \ldots |\sigma''| - 1\}$.
  Since $P'' \leq P$ we find $\text{pos}''^{-1}$ to be an injection $\{0 \ldots |\sigma''| - 1\} \into \TT$.
  As $\sigma''$ is finite, $j_{\text{max}} := \text{max}_{i \in \{0 \ldots |\sigma''| - 1\}} \text{pos}(\text{pos}''^{-1}(i))$ exists.
  Let $\sigma_2$ be the prefix of $\sigma$ of length $j_{\text{max}}+1$.
  Then $\text{pos}^{-1}(\sigma_2)$ selects a set of transitions from $\TT$ which together with
  the connecting places forms a prefix $P_2 \leq P$ (cf.\ \reflem{prefix-down-fs}).
  Let $P_2 = ((\SS_2, \TT_2, \FF_2, \MM_0), \pi_2)$.
  As $\sigma_2$ was chosen long enough, we find
  $\text{pos}\circ\text{pos}''^{-1}$ to be an injection from $\{0 \ldots |\sigma''| - 1\}$ into
  $\{0 \ldots |\sigma_2| - 1\}$ and hence
  $\text{pos}^{-1}\circ\text{pos}\circ\text{pos}''^{-1}$ to be an injection not only into $\TT$ but
  also into just $\TT_2$.
  Clearly then $\TT'' \subseteq \TT_2$.
  Also $\pi'' = \pi \restrictedto (\SS'' \cup \TT'') = (\pi \restrictedto (\SS_2 \cup \TT_2)) \restrictedto (\SS'' \cup \TT'') =
  \pi_2 \restrictedto (\SS'' \cup \TT'')$ which is to say, since both $\sigma''$ and $\sigma_2$ select some prefix from the
  beginning of the same $P$, they must have the same structure between shared transitions.
  Hence $P'' \leq P_2$.
  From \reflem{prefix-up-GR} we then obtain a $\sigma_1 \in \Lin(P_2)$ with $\sigma'' \leq \sigma_1$.
  As $\sigma_1 \in \Lin(P_2)$ and $\sigma_2 \in \Lin(P_2)$ we conclude, using Theorem~7.10 from \cite{best87both}, that $\sigma_1 \equiv_0^* \sigma_2$.
\qed
\end{proof}

\noindent
Besides these lemmas, we only need the following ``finitary'' version of \refthm{BD-infty}.

\begin{proposition}\rm\label{pr-BD}
Let $\sigma,\rho\in\fFS$ and $P,Q \in \fGR$ such that
$\sigma \in \Lin(P)$ and $\rho\in\Lin(Q)$.
Then $\sigma \equiv_0^* \rho$ iff $P \equiv_1^* Q$.
\end{proposition}

\begin{proof}
  In \cite{glabbeek11ipl} as Theorem 3 and an immediate conclusion from two theorems of \cite{best87both}.
\end{proof}

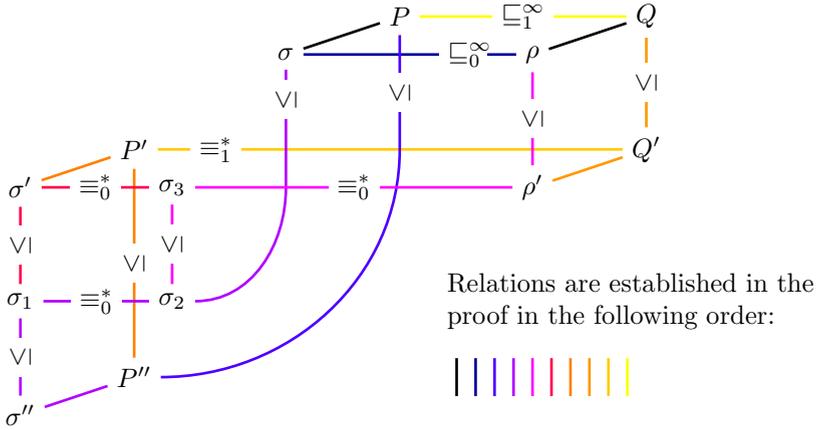
\definecolor{lsA}{rgb}{0.0,0.0,0.0}\def\lsA{[color=lsA,text=black,line width=1pt]}
\definecolor{lsB}{rgb}{0.0,0.0,0.6}\def\lsB{[color=lsB,text=black,line width=1pt]}
\definecolor{lsC}{rgb}{0.3,0.0,1.0}\def\lsC{[color=lsC,text=black,line width=1pt]}
\definecolor{lsD}{rgb}{0.7,0.0,1.0}\def\lsD{[color=lsD,text=black,line width=1pt]}
\definecolor{lsE}{rgb}{1.0,0.0,1.0}\def\lsE{[color=lsE,text=black,line width=1pt]}
\definecolor{lsF}{rgb}{1.0,0.0,0.3}\def\lsF{[color=lsF,text=black,line width=1pt]}
\definecolor{lsG}{rgb}{1.0,0.5,0.0}\def\lsG{[color=lsG,text=black,line width=1pt]}
\definecolor{lsH}{rgb}{1.0,0.6,0.0}\def\lsH{[color=lsH,text=black,line width=1pt]}
\definecolor{lsI}{rgb}{1.0,0.8,0.0}\def\lsI{[color=lsI,text=black,line width=1pt]}
\definecolor{lsJ}{rgb}{1.0,1.0,0.0}\def\lsJ{[color=lsJ,text=black,line width=1pt]}

\begin{center}
\begin{tikzpicture}
  \draw (0, 1) node(sxx) {$\sigma''$}
        (0, 2.5) node(s1)  {$\sigma_1$}
        (0, 4) node(sx)  {$\sigma'$}
        (2, 2.5) node(s2)  {$\sigma_2$}
        (2, 4) node(s3)  {$\sigma_3$}
        (6.75, 4) node(rx)  {$\rho'$}
        (6.75, 5.75) node(r)   {$\rho$}
        (3.5, 5.75) node(s)   {$\sigma$}

        (8.25, 6.25) node(Q)   {$Q$}
        (5.0, 6.25) node(P)   {$P$}
        (8.25, 4.5) node(QX)  {$Q'$}
        (1.5, 4.5) node(PX)  {$P'$}
        (1.5, 1.5) node(PXX) {$P''$}
        ;

  \draw \lsG (PXX) -- (PX)             node [sloped,midway,fill=white] {$\leq$} ;
  \draw \lsH (QX)  -- (Q)              node [sloped,midway,fill=white] {$\leq$} ;
  \draw \lsI (PX)  -- (QX)             node [sloped,very near start,fill=white] {$\equiv_1^*$} ;
  \draw \lsJ (P)   -- (Q)              node [sloped,midway,fill=white] {$\sqsubseteq_1^\infty$} ;
  \draw \lsB (s)   -- (r)              node [sloped,near end,fill=white] {$\sqsubseteq_0^\infty\!\!\!$} ;

  \draw \lsD (sxx) -- (s1)             node [sloped,midway,fill=white] {$\leq$} ;
  \draw \lsF (s1)  -- (sx)             node [sloped,midway,fill=white] {$\leq$} ;
  \draw \lsE (s2)  -- (s3)             node [sloped,midway,fill=white] {$\leq$} ;
  \draw \lsE (rx)  -- (r)              node [sloped,midway,fill=white] {$\leq$} ;
  \draw \lsD (s2)  to[out=0,in=270] (3.5,4) -- (s)   node [sloped,near end,fill=white] {$\leq$}  ;
  \draw \lsC (PXX) to[out=0,in=270] (5.0,4.5) -- (P) node [sloped,midway,fill=white] {$\leq$} ;
  \draw \lsD (s1)  -- (s2)             node [sloped,midway,fill=white] {$\equiv_0^*$} ;
  \draw \lsF (sx)  -- (s3)             node [sloped,midway,fill=white] {$\equiv_0^*$} ;
  \draw \lsE (s3)  -- (rx)             node [sloped,midway,fill=white] {$\equiv_0^*$} ;

  \draw \lsD (sxx) -- (PXX) ;
  \draw \lsG (sx)  -- (PX) ;
  \draw \lsH (rx)  -- (QX) ;
  \draw \lsA (s)   -- (P) ;
  \draw \lsA (r)   -- (Q) ;

  \draw (5.5, 2.5) node[text width=5cm, right]
    {Relations are established in the proof in the following order:} ;
  \draw \lsA (5.75,1.75) -- (5.75,1.25) ;
  \draw \lsB (6.00,1.75) -- (6.00,1.25) ;
  \draw \lsC (6.25,1.75) -- (6.25,1.25) ;
  \draw \lsD (6.50,1.75) -- (6.50,1.25) ;
  \draw \lsE (6.75,1.75) -- (6.75,1.25) ;
  \draw \lsF (7.00,1.75) -- (7.00,1.25) ;
  \draw \lsG (7.25,1.75) -- (7.25,1.25) ;
  \draw \lsH (7.50,1.75) -- (7.50,1.25) ;
  \draw \lsI (7.75,1.75) -- (7.75,1.25) ;
  \draw \lsJ (8.00,1.75) -- (8.00,1.25) ;
\end{tikzpicture}
\end{center}

\noindent
\textit{Proof of \refthm{ordering}:}
``Only if'': Suppose $\sigma \sqsubseteq_0^\infty \rho$.
Let $P''\in\fGR$ with $P''\leq P$. It suffices to show that there are
$P',Q'\in\fGR$ with $P'' \leq P' \equiv_1^* Q' \leq Q$.
By \reflem{prefix-down-GR} there are $\sigma''\in\Lin(P'')$
and $\sigma_1,\sigma_2\in\fFS$ with
$\sigma'' \leq \sigma_1 \equiv_0^* \sigma_2 \leq \sigma$.\linebreak[3]
By \refdf{BD-swapping-fs-alt}, using that $\sigma \sqsubseteq_0^\infty \rho$
and $\sigma_2\leq \sigma$, there are $\sigma_3,\rho'\in\fFS$ with
$\sigma_2 \leq \sigma_3 \equiv_0^* \rho' \leq \rho$.
By \refcor{swapprefix-fs}, using that $\sigma_1 \connectedto \sigma_2
\leq \sigma_3$, there is a $\sigma'\in\fFS$ with
$\sigma_1 \leq \sigma' \equiv_0^* \sigma_3$.
Hence $\sigma'' \leq \sigma' \equiv_0^* \rho' \leq \rho$
by the transitivity of $\leq$ and $\equiv_0^*$.
By \reflem{prefix-up-fs}, using that $\sigma'' \leq \sigma'$ and
$\sigma''\in\Lin(P'')$, there is a $P'\mathbin\in\fGR$ with
$\sigma'\in\Lin(P')$ and $P''\leq P'$.
By \reflem{prefix-down-fs}, substituting $Q$, $\rho$ and $\rho'$
for $P$, $\sigma$ and $\sigma''$, and using that $\rho'\leq \rho$
and $\rho\in\Lin(Q)$, there is a $Q'\inp\fGR$ with $\rho'\inp\Lin(Q')$
and $Q'\leq Q$.
By \refpr{BD}, using that $\sigma'\in\Lin(P')$, $\rho'\in\Lin(Q')$ and
$\sigma'\equiv_0^* \rho'$, we conclude $P' \equiv_1^* Q'$.

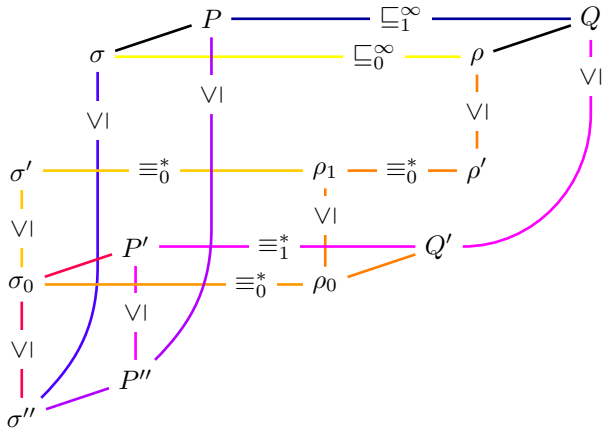
\begin{center}
\begin{tikzpicture}
  \draw (0, 0.5) node(sxx) {$\sigma''$}
        (0, 2.25) node(s0) {$\sigma_0$}
        (0, 3.75) node(sx) {$\sigma'$}
        (4, 2.25) node(r0) {$\rho_0$}
        (4, 3.75) node(r1) {$\rho_1$}
        (6, 3.75) node(rx) {$\rho'$}
        (1, 5.25) node(s)  {$\sigma$}
        (6, 5.25) node(r)  {$\rho$}

        (1.5,1) node(PXX) {$P''$}
        (1.5,2.75) node(PX) {$P'$}
        (5.5,2.75) node(QX) {$Q'$}
        (2.5,5.75) node(P) {$P$}
        (7.5,5.75) node(Q) {$Q$}
        ;

  \draw \lsF (sxx) -- (s0) node [sloped,midway,fill=white] {$\leq$} ;
  \draw \lsI (s0) -- (sx) node [sloped,midway,fill=white] {$\leq$} ;
  \draw \lsE (PXX) -- (PX) node [sloped,midway,fill=white] {$\leq$} ;
  \draw \lsE (QX) to[out=0,in=270] (7.5,4.5) -- (Q) node [sloped,midway,fill=white] {$\leq$} ;
  \draw \lsG (r0) -- (r1) node [sloped,pos=0.65,fill=white] {$\leq$} ;
  \draw \lsG (rx) -- (r) node [sloped,midway,fill=white] {$\leq$} ;
  \draw \lsC (sxx) to[out=45,in=270] (1,2.5) -- (s) node [sloped,near end,fill=white] {$\leq$} ;
  \draw \lsH (s0) -- (r0) node [sloped,pos=0.8,fill=white] {$\equiv_0^*$} ;
  \draw \lsI (sx) -- (r1) node [sloped,pos=0.43,fill=white] {$\equiv_0^*$} ;
  \draw \lsG (r1) -- (rx) node [sloped,midway,fill=white] {$\equiv_0^*$} ;
  \draw \lsE (PX) -- (QX) node [sloped,pos=0.45,fill=white] {$\equiv_1^*$} ;

  \draw \lsB (P) -- (Q) node [sloped,midway,fill=white] {$\sqsubseteq_1^\infty$} ;
  \draw \lsJ (s) -- (r) node [sloped,near end,fill=white] {$\sqsubseteq_0^\infty$} ;

  \draw \lsD (sxx) -- (PXX) ;
  \draw \lsA (r) -- (Q) ;
  \draw \lsA (s) -- (P) ;
  \draw \lsF (s0) -- (PX) ;
  \draw \lsG (r0) -- (QX) ;
  \draw \lsD (PXX) to[out=45,in=270] (2.5,3) -- (P) node [sloped,pos=0.7,fill=white] {$\leq$} ;

  % \draw (4.5, 1.5) node[text width=5cm, right]
  %   {relations are established in the proof in the following order:} ;
  % \draw \lsA (4.75,0.75) -- (4.75,0.25) ;
  % \draw \lsB (5.00,0.75) -- (5.00,0.25) ;
  % \draw \lsC (5.25,0.75) -- (5.25,0.25) ;
  % \draw \lsD (5.50,0.75) -- (5.50,0.25) ;
  % \draw \lsE (5.75,0.75) -- (5.75,0.25) ;
  % \draw \lsF (6.00,0.75) -- (6.00,0.25) ;
  % \draw \lsG (6.25,0.75) -- (6.25,0.25) ;
  % \draw \lsH (6.50,0.75) -- (6.50,0.25) ;
  % \draw \lsI (6.75,0.75) -- (6.75,0.25) ;
  % \draw \lsJ (7.00,0.75) -- (7.00,0.25) ;
\end{tikzpicture}
\end{center}

``If'': Suppose $P \sqsubseteq_1^\infty Q$.
Let $\sigma''\in\fFS$ with $\sigma''\leq \sigma$. It suffices to show
that there are $\sigma',\rho'\in\fFS$ with
$\sigma'' \leq \sigma' \equiv_0^* \rho' \leq \rho$.
By \reflem{prefix-down-fs}, using that $\sigma''\leq
\sigma\in\Lin(P)$, there is a $P''\in\fGR$ with
$\sigma''\in\Lin(P'')$ and $P''\leq P$.
By \refdf{BD-swapping-alt}, using that $P \sqsubseteq_1^\infty Q$
and $P''\leq P$, there are $P',Q'\in\fGR$ with
$P'' \leq P' \equiv_1^* Q' \leq Q$.
By \reflem{prefix-up-GR}, using that $P''\leq P'$ and
$\sigma''\in\Lin(P'')$, there is a $\sigma_0\mathbin\in\Lin(P')$ with
$\sigma''\leq \sigma_0$.
By \reflem{prefix-down-GR}, substituting $Q'$, $Q$ and $\rho$ for
$P''$, $P$ and $\sigma$, and using that $Q'\leq Q$ and
$\rho\in\Lin(Q)$, there are $\rho_0\mathbin\in\Lin(Q')$ and
$\rho_1,\rho'\in\fFS$ with
$\rho_0 \leq \rho_1 \equiv_0^* \rho' \leq \rho$.
By \refpr{BD}, using that $\sigma_0\in\Lin(P')$, $\rho_0\in\Lin(Q')$ and
$P' \equiv_1^* Q'$, we obtain $\sigma_0 \equiv_0^* \rho_0$.
By \refcor{swapprefix-fs}, using 
$\sigma_0 \connectedto \rho_0 \leq \rho_1$,
there is a $\sigma'\in\fFS$ with
$\sigma_0 \leq \sigma' \equiv_0^* \rho_1$. Hence
$\sigma'' \leq \sigma' \equiv_0^* \rho' \leq \rho$
by the transitivity of $\leq$ and $\equiv_0^*$.
\qed

\section{A countable conflict-free net has a largest process}
\label{sec-uniquemaxprocess}

We now show that a countable conflict-free net has a largest process. As we have
an order-preserving bijection between FS-process or BD-process,
it does not matter which notion of process we use here.
We prove an even stronger result, using binary-conflict-\!-free instead of conflict-free.
In preparation we need the following lemmas.

\begin{lemma}\rm\label{lem-conflictfreeswap}
  Let $N = (S, T, F, M_0)$ be a binary-conflict-\!-free net,
  $\sigma t, \sigma u \inp \FS(N)$ with $\sigma\inp T^*$, $t, u \inp T$, and $t \mathbin{\ne} u$.

  Then $\sigma tu, \sigma ut \in \FS(N)$ and
  $\sigma tu \connectedto \sigma ut$.
\end{lemma}
\begin{proof}
  As we have unlabelled transitions, $\sigma$ leads to a unique marking.
  From $M_0 \goesto[\sigma t]{} \wedge M_0 \goesto[\sigma u]$ we thus
  have that an $M_1$ exists with $M_0 \goesto[\sigma]{} M_1 \wedge
  M_1 {\goesto[t]} \wedge M_1 \goesto[u]$. Due to binary-conflict-\!-freeness
  then also $M_1 \goesto[\{t, u\}]$.
  Hence $M_0 \goesto[\sigma]\goesto[\{t, u\}]$, so
  $\sigma tu, \sigma ut \in \FS(N)$ and
  $\sigma tu \connectedto \sigma ut$.
  \qed
\end{proof}

\begin{lemma}\rm\label{lem-closingdiamondwithout}
  Let $N = (S, T, F, M_0)$ be a binary-conflict-\!-free net,
  $\sigma t, \sigma \rho \inp \FS(N)$ with $t \inp T$, $\sigma,\rho\in T^*$,
  and $t \mathbin{\notin} \rho$.

  Then $\sigma t\rho, \sigma\rho t \in \FS(N)$ and
  $\sigma t\rho \connectedto \sigma \rho t$.
\end{lemma}
\begin{proof}
  Via induction on the length of $\rho$.

  If $\rho \mathbin= \epsilon$, $\sigma t\in\FS(N)$ trivially implies
  $\sigma \epsilon t, \sigma t \epsilon\in\FS(N)$ and
  $\sigma \epsilon t \connectedto \sigma t \epsilon$.

  For the induction step take $\rho \mathbin{:=} u\mu$ (thus $u \ne t$).
  With $\sigma t, \sigma u\in\FS(N)$ and \reflem{conflictfreeswap} also
  $\sigma u t\in\FS(N)$ and $\sigma t u \connectedto \sigma u t$.
  Together with $\sigma u \mu\in\FS(N)$, the induction assumption then gives us
  $\sigma u t \mu \in \FS(N)$ and
  $\sigma u t \mu \connectedto \sigma u \mu t =
  \sigma \rho t$. With $\sigma u t \connectedto \sigma t u$ also
  $\sigma u t \mu \connectedto \sigma t u \mu =
  \sigma t \rho$ and $\sigma\rho t,\,\sigma t \rho\in \FS(N)$.
  \qed
\end{proof}

\begin{lemma}\rm\label{lem-closingdiamondwith}
  Let $N = (S, T, F, M_0)$ be a binary-conflict-\!-free net,
  $\sigma, \rho_1, \rho_2 \in T^*$, $t \in T$, $t \notin \rho_1$.

  If $\sigma t \in\FS(N) \wedge \sigma \rho_1 t \rho_2 \in\FS(N)$
  then $\sigma t \rho_1\rho_2\in\FS(N) \wedge
  \sigma t \rho_1 \rho_2 \connectedto \sigma \rho_1 t \rho_2$.
\end{lemma}
\begin{proof}
  Applying \reflem{closingdiamondwithout} with
  $\sigma t\inp\FS(N) \wedge \sigma \rho_1\inp\FS(N)$ 
  we get $\sigma t \rho_1\inp\FS(N) \wedge
  \sigma t \rho_1 \connectedto \sigma \rho_1 t$.
  Since $\sigma \rho_1 t \rho_2\in\FS(N)$ the latter yields
  $\sigma t \rho_1 \rho_2 \connectedto \sigma \rho_1 t \rho_2$
  and thus $\sigma t \rho_1 \rho_2\in\FS(N)$.
  \qed
\end{proof}

\begin{lemma}\rm\label{lem-closingdiamond}
  Let $N$ be a binary-conflict-\!-free net.
  \vspace{2pt}

  If $\sigma,\sigma'\inp\FS(N)$ then
  $\exists \mu, \mu'. \sigma \mu\inp\FS(N) \wedge
  \sigma' \mu'\in\FS(N) \wedge \sigma\mu \connectedto \sigma'\mu'$.
\end{lemma}
\begin{proof}
  Via induction on the length of $\sigma$.

  If $\sigma = \epsilon$ we take $\mu = \sigma'$ and $\mu' = \epsilon$.

  For the induction step we start with
  $$\sigma, \sigma'\in\FS(N) \implies
  \exists \mu, \mu'. \sigma\mu\in\FS(N) \wedge \sigma'\mu'\in\FS(N)
  \wedge \sigma\mu \connectedto \sigma'\mu'$$ and need to show that
  $$\sigma t, \sigma'\in\FS(N) \implies
  \exists \bar{\mu}, \bar{\mu}'.
  \sigma t\bar{\mu}\in\FS(N) \wedge \sigma'\bar{\mu}'\in\FS(N)
  \wedge \sigma t\bar{\mu} \connectedto \sigma'\bar{\mu}'\trail{.}$$

  If $t \inp \mu$, $\mu$ must be of the form $\mu_1 t \mu_2$ with $t \notin \mu_1$.
  We then take $\bar{\mu} := \mu_1 \mu_2$ and $\bar{\mu}' := \mu'$.
  By \reflem{closingdiamondwith} we find
  $\sigma t \mu_1 \mu_2\in\FS(N)$, i.e.\
  $\sigma t \bar{\mu}\in\FS(N)$.
  By the induction assumption
  $\sigma' \bar{\mu}'\in\FS(N)$.
  Per \reflem{closingdiamondwith} $\sigma t \bar{\mu} =
  \sigma t \mu_1\mu_2 \connectedto \sigma \mu_1 t \mu_2 =
  \sigma \mu$. From the induction assumption we obtain
  $\sigma \mu \connectedto \sigma' \mu' = \sigma' \bar{\mu}'$.

  If $t \mathbin{\notin} \mu$, we take $\bar{\mu} := \mu$ and $\bar{\mu}' := \mu't$.
  By \reflem{closingdiamondwithout} we find that
  $\sigma t \mu, \sigma \mu t\in\FS(N)$, i.e.\ also
  $\sigma t \bar{\mu}\in\FS(N)$. From $\sigma \mu t\in\FS(N)$
  and $\sigma \mu \connectedto \sigma' \mu'$ follows that
  $\sigma' \mu' t\in\FS(N)$, i.e.\ $\sigma' \bar{\mu}'\in\FS(N)$.
  Also by \reflem{closingdiamondwithout} we find
  $\sigma t \bar{\mu} = \sigma t \mu \connectedto \sigma \mu t$.
  From the induction assumption  we obtain
  $\sigma \mu t \connectedto \sigma' \mu' t = \sigma' \bar{\mu}'$.
  \qed
\end{proof}

\begin{theorem}\rm\label{thm-fsprocessexists}
  Let $N=(S,T,F,M_0)$
  be a countable, binary-conflict-\!-free net.
  
  Then $N$ has a {$\sqsubseteq_0^\infty$}-largest FS-process.
\end{theorem}

\begin{proof}
  Since $N$ is countable, so is the set $\fFS$ of its finite firing sequences.
  Enumerate its elements as $\sigma_1$, $\sigma_2$, $\ldots$.
  
  By induction, we will construct two sequences $\rho_1, \rho_2, \ldots$ and $\sigma'_1$, $\sigma'_2$, $\ldots$
  of finite firing sequences, such that, for all $i>0$,
  (1)  $\rho_i \leq \rho_{i+1}$, and
  (2)  $\sigma_i \leq \sigma'_i  \equiv^*_0 \rho_i$.\linebreak
  Now let $\rho \in \FS^\infty(N)$ be the limit of all the $\rho_i$.
  As $\sqsubseteq^\infty_0$ is defined in terms of finite prefixes,
  $\sigma \sqsubseteq^\infty_0 \rho$ for any $\sigma \in \FS^\infty(N)$,
  so that $\rho$ is the {$\sqsubseteq_0^\infty$}-largest FS-process of $N$.
  \vspace{1ex}

  \noindent
  \emph{Induction base:} Take $\rho_1 := \sigma'_1 := \sigma_1$.
  \vspace{1ex}

  \noindent
  \emph{Induction step:} Given $\rho_i$,
  by \reflem{closingdiamond} there are $\mu,\mu' \in T^*$ such that $\rho_i \mu$ and
  $\sigma_{i+1} \mu'\in\FS(N)$ and $\rho_i\mu \connectedto \sigma_{i+1}\mu'$.
  Take $\rho_{i+1} = \rho_i \mu$ and $\sigma'_{i+1} := \sigma_{i+1}\mu'$.
  \qed
\end{proof}

\begin{corollary}\rm\label{cor-onemaximalprocess1}
  A countable and binary-conflict-\!-free net $N$ has exactly one \linebreak $\sqsubseteq^\infty_1$-largest
  BD-process.
\end{corollary}
\begin{proof}
  By \refthm{fsprocessexists}, $N$ has a $\sqsubseteq^\infty_0$-largest FS-process.
  Take any representative firing sequence $\sigma$ thereof.
  By \refpr{countable} there is a $P \mathbin\in \GR$ with $\sigma \mathbin\in \Lin(P)$.

  Now take any $Q \in \GR$.
  As $N$ is countable, so is $Q$.
  From \refpr{countable} thus exists $\rho \in \Lin(Q)$.
  As $\sigma$ comes from the largest FS-process, $\rho \sqsubseteq^\infty_0 \sigma$.
  By \refthm{ordering} then $Q \sqsubseteq^\infty_1 P$.

  Thus $P$ is a representative of the largest BD-process of $N$.
  \qed
\end{proof}

\noindent
\refcor{onemaximalprocess1} does not hold for uncountable nets,
as witnessed by the counterexample in \reffig{uncountable}.
This binary-conflict-\!-free net $N$ has a transition $t$ for each real number $t\inp\bbbr$.
Each such transition has a private preplace $s_t$ with $M_0(s_t)=1$
and $F(s_t,t)=1$, which ensures that $t$ can fire only
once. Furthermore there is one shared place $s$ with $M_0(s)=2$ and a
loop $F(s,t)=F(t,s)=1$ for each transition $t$. There are no other
places, transitions or arcs besides the ones mentioned above.

Each GR-process of $N$, and hence also each BD-process $P$, has only
countably many transitions. Moreover, any two GR-processes firing the same
countable set of transitions of $N$ are swapping equivalent. Thus a BD-process is fully determined
by a countable set of reals, and the $\sqsubseteq^\infty_1$-order between BD-processes corresponds with set-inclusion.
It follows that $N$ does not have a  $\sqsubseteq^\infty_1$-largest BD-process.

\begin{figure}
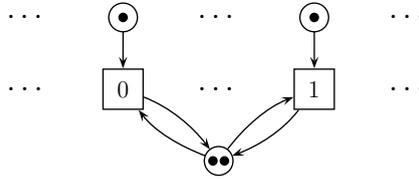

\vspace{-2ex}
  \begin{center}
    \psscalebox{0.9}{
    \begin{petrinet}(8,4)
      \P(2,3.5):p0;
      \P(6,3.5):p1;

      \p(4,0.5):pall;
      \pscircle*(3.88,0.5){0.1}
      \pscircle*(4.12,0.5){0.1}

      \t(2,2):t0:$0$;
      \t(6,2):t1:$1$;

      \a p0->t0;
      \a p1->t1;

      \A pall->t0; \A t0->pall;
      \A pall->t1; \A t1->pall;
      
      \rput(0,3.5){\psscalebox{2}{$\cdots$}}
      \rput(8,3.5){\psscalebox{2}{$\cdots$}}
      \rput(4,3.5){\psscalebox{2}{$\cdots$}}
      \rput(0,2.0){\psscalebox{2}{$\cdots$}}
      \rput(8,2.0){\psscalebox{2}{$\cdots$}}
      \rput(4,2.0){\psscalebox{2}{$\cdots$}}
    \end{petrinet}
    }
  \end{center}
  \vspace{-3ex}
  \caption{A net without a  $\sqsubseteq^\infty_1$-largest BD-process.}
  \label{fig-uncountable}
  \vspace{-2ex}
\end{figure}

\section{Conclusion}

Best and Devillers \cite{best87both} established a bijective correspondence between
BD-processes and FS-processes (our terminology) of countable place/transition systems.
A BD-process is an equivalence class of Goltz-Reisig processes under the notion of swapping
equivalence proposed in \cite{best87both}.
An FS-process is an equivalence class of firing sequences under a related notion of
equivalence also proposed in \cite{best87both}.
Here we considered natural partial orders on BD-processes as well as on FS-processes, and showed that
the bijective correspondence between BD- and FS-processes preserves these orders, and hence the
notion of a largest process.

Moreover, we showed that a countable place/transition system without binary conflicts has a largest
FS-process, and hence a largest BD-process.
By means of a counterexample we indicated that this result does not extend to uncountable nets.

We showed in \cite{glabbeek11ipl} that the reverse direction, that a place/transition system with a largest BD-process is binary-conflict-\!-free,
holds for a large class of Petri nets, called structural
conflict nets, which include the safe nets. The example from \reffig{badswapping} shows it does not
hold for arbitrary countable place/transition systems. This system has a largest BD-process but
does have a binary conflict: after the $a$-transition, both $b$ and and $c$ are possible, but
the step $\{b,c\}$ is not.

The question whether an uncountable net without (any) conflict always has a
largest BD-process is left open.

\bibliographystyle{eptcsalpha}

\end{document}